\documentclass[11pt]{article}
\usepackage[a4paper, total={6in, 8in}]{geometry}
\usepackage{enumerate}
\usepackage{amsmath}
\usepackage{amssymb,latexsym}
\usepackage{amsthm}
\usepackage{color}
\usepackage{cancel}
\usepackage{graphicx}
\usepackage{cite}
\usepackage{changes}
\usepackage{lineno}
\usepackage{hyperref}
\usepackage{comment}
\usepackage{amscd}
\usepackage{cleveref}

\DeclareMathOperator{\C}{\mathcal{C}}

\DeclareMathOperator{\supp}{supp}

\newtheorem{theorem}{Theorem}[section]

\newtheorem{lemma}[theorem]{Lemma}
\newtheorem{corollary}[theorem]{Corollary}
\newtheorem{definition}[theorem]{Definition}
\newtheorem{proposition}[theorem]{Proposition}

\newtheorem{example}[theorem]{Example}

\newtheorem{remark}[theorem]{Remark}

\newcommand{\fqm}{\mathbb{F}_{q^m}}

\newcommand{\F}{{\mathbb F}}

\newcommand{\w}{{\mathrm w}}

\newcommand{\GL}{\hbox{{\rm GL}}}

\newcommand{\fq}{{\mathbb F}_{q}}

\newcommand{\N}{\mathrm{N}}

\newcommand{\uu}{\mathbf{u}}
\newcommand{\cc}{\mathbf{c}}
\newcommand{\xx}{\mathbf{x}}

\newcommand{\ee}{\mathbf{e}}

\newcommand{\zz}{\mathbf{z}}
\newcommand{\vv}{\mathbf{v}}

\newcommand{\nkd}{[n,k,d]_{q^m/q}}
\newcommand{\nk}{[n,k]_{q^m/q}}

\title{Completely decomposable rank-metric codes}
\author{Paolo Santonastaso}
\date{}

\begin{document}

\maketitle

\begin{abstract}
In this paper, we investigate completely decomposable rank-metric codes, i.e. rank-metric codes that are the direct sum of 1-dimensional maximum rank distance codes. We study the weight distribution of such codes, characterizing codewords with certain rank weights. Additionally, we obtain classification results for codes with the largest number of minimum weight codewords within the class of completely decomposable codes.
\end{abstract}

%\tableofcontents


\section{Introduction}
Rank-metric codes have become increasingly popular in the recent years due to their several applications and because of their own interest. These codes have indeed been studied in connection with several algebraic and combinatorial objects, such as semifields, linear sets, tensorial algebras, and skew algebras, see e.g. \cite{sheekeysurvey}. Such codes can be seen as subspaces of matrices over a field, namely $\mathbb{K},$ or as subspaces of vectors of length $n$ over a degree $m$ extension of $\mathbb{K}$.
In this paper, we investigate special classes of codes that are the (external) direct sum of two or more codes, i.e., \emph{decomposable codes}. The first application of the direct sum of rank-metric codes is discussed in \cite{gabidulin2005subcodes}, dealing with the decoding problem beyond the error-correcting capability of the code. Rank-metric codes have been widely used in network coding, see e.g. \cite{bartz2022rank} and in the schemes proposed in \cite{silva2008rank,silva2011universal}, the authors employ direct sum of maximum rank distance (MRD) codes.

Decomposable codes have also been explored in the literature in other metrics. For instance, in the Hamming metric, classes of efficiently decodable decomposable codes are studied in \cite{takata1994suboptimum,he2023error}. In \cite{munuera1993equality}, some properties on the structure of these codes are examined in the context of algebraic geometric codes.

In this paper, we focus on \textbf{completely decomposable rank-metric codes}, i.e., rank-metric codes $\C$ that are, up to equivalence, the (external) direct sum of $1$-dimensional MRD codes $\C_1,\ldots,\C_k$. We characterize completely decomposable rank-metric codes as those codes that admit a basis ${\mathbf{c}_1, \ldots, \mathbf{c}_k}$ such that the sum of the weights of the $\mathbf{c}_i$'s equals the length of the code; see \Cref{th:characterization1dim}. This property is also expressed in terms of the supports of the codewords of such a basis. Moreover, we show that completely decomposable codes are not minimal rank-metric codes and we can detect the set of minimal codewords, see \Cref{prop:descriptionminimal}. We also relate the study of the weight distribution of a completely decomposable rank-metric code $\C$ to some of its shortened and punctured codes. In particular, we characterize codewords with certain rank weights, see \Cref{lm:numberpointsweightcomplementary} and \Cref{th:numberpointsweightcomplementary}. Furthermore, we obtain classification results for codes with the maximum number of minimum weight codewords within the class of completely decomposable codes, see \Cref{th:charminimumweightnotprime} and \Cref{th:charshortnedpunct}.

\section{Preliminaries} \label{sec:prel}

In this section, we first recall some properties of $\F_q$-subpaces of a finite dimensional $\F_{q^m}$-vector space under the action of a nondegenerate reflexive bilinear form over $\F_q$ induced by a nondegenerate reflexive bilinear form over $\F_{q^m}$. Then, in Subsection \ref{subsec:product}, we study multiplicative properties of $\F_q$-subspaces of $\F_{q^m}$. Finally, in Subsection \ref{subsec:rank}, we recall the results on rank-metric codes and their geometry. In the following, $q$ is a prime power and $\mathbb{F}_q$ is the finite field with $q$ elements.

\subsection{Dual of a subspace}

We recall the notion of the dual of an $\F_q$-subspace of the $k$-dimensional $\F_{q^m}$-vector space $V$.
Let $\sigma \colon V \times V \rightarrow \F_{q^m}$ be a nondegenerate reflexive bilinear form defined on $V$ and consider \[
\begin{array}{cccc}
    \sigma': & V \times V & \longrightarrow & \F_q  \\
     & (x,y) & \longmapsto & \mathrm{Tr}_{q^m/q} (\sigma(x,y)),
\end{array}
\] 
where $\mathrm{Tr}_{q^m/q}(x)=x+x^q+\ldots+x^{q^{m-1}}$ denotes the trace function.
In this way, $\sigma'$ turns out to be a nondegenerate reflexive bilinear form on $V$ seen as an $\fq$-vector space of dimension $km$. Consider $\perp$ and $\perp'$ as the orthogonal complement maps defined by $\sigma$ and $\sigma'$, respectively. For an $\F_q$-subspace $U$ of $V$, the \textbf{dual subspace of $U$} with respect to $\sigma'$ is the $\F_q$-subspace $U^{\perp'}$ in $V$. If $\perp_1'$ and $\perp_2'$ are orthogonal complement maps defined by the nondegenerate reflexive bilinear forms $\sigma_1$ and $\sigma_2$, respectively, then the dual subspaces $U^{\perp_1'}$ and $U^{\perp_2'}$ of $U$ are $\mathrm{GL}(k,q^m)$-equivalent, cf. \cite[Proposition 2.5]{polverino2010linear}.

Moreover, for any $\F_q$-subspace $U$ of $V$, we have $\dim_{\fq}(U^{\perp'})=km-\dim_{\fq}(U)$ and for any $\F_{q^m}$-subspace of $V$, it holds that $W^{\perp'}=W^{\perp}$. As a consequence, the following relation holds.

\begin{proposition} [see \textnormal{\cite[Property 2.6]{polverino2010linear}}] \label{prop:weightdual}
Let $U$ be an $\fq$-subspace of $V$ and $W$ be an $\fqm$-subspace of $V$.
Then 
\[ \dim_{\fq}(U^{\perp'}\cap W^{\perp})=\dim_{\fq}(U\cap W)+\dim_{\fq}(V)-\dim_{\fq}(U)-\dim_{\fq}(W). \]
\end{proposition}

\noindent In the next, we will always consider $V=\F_{q^m}^k$ and $\sigma$ as the standard inner product
\[ \sigma \colon ((u_1,\ldots,u_k),(v_1,\ldots,v_k)) \in (\mathbb{F}_{q^m}^k)^2 \mapsto u_1v_1+\ldots+u_kv_k \in \F_{q^m}, \]
and hence
\[ \sigma' \colon ((u_1,\ldots,u_k),(v_1,\ldots,v_k)) \in (\mathbb{F}_{q^m}^k)^2 \mapsto \mathrm{Tr}_{q^m/q}(u_1v_1+\ldots+u_kv_k) \in \F_{q}. \]

Note that, when $k=1$, $\sigma'=\mathrm{Tr}_{q^m/q}(\cdot)$ is a nondegenerate reflexive bilinear form on $\F_{q^m}$ over $\F_q$ and in what follows, if $W$ is an $\F_q$-subspace of $\F_{q^m}$, we define \[
W^{\perp^*}=\{b \in \F_{q^m} \colon \mathrm{Tr}_{q^m/q}(ab)=0, \mbox{ for every }a \in W\},\]
the ortogonal complement of $W$ with respect to $\mathrm{Tr}_{q^m/q}(\cdot)$.

If $\F_{q^e}$ is a subfield of $\F_{q^m}$, the form 
\[(u,v) \in \F_{q^m}^2 \mapsto \mathrm{Tr}_{q^m/q^e}(uv)\] is a nondegenerate reflexive bilinear form on $\F_{q^m}$ over $\F_{q^e}$, as well. So, in the following, when it is necessary, we use the symbol $\perp^*_{q^m/q^e}$ to denote the orthogonal complement defined by $\mathrm{Tr}_{q^m/q^e}$ on $\F_{q^m}$ over $\F_{q^e}$.

When $W$ is generated by a subsequence of a geometric progression of a generator $\lambda$ of $\F_{q^m}$, $W^{\perp^*}$ is as follows.

\begin{proposition} [see \textnormal{\cite[Proposition 2.8]{napolitano2023classifications}}] \label{prop:dualwithdual}
Let $\lambda \in \F_{q^m}$ be such $\F_q(\lambda)=\F_{q^m}$. Let $W=\langle 1,\lambda,\ldots,\lambda^{t-1} \rangle_{\F_q}$, with $t \in \{1,\ldots,m-1\}$. Then $ W^{\perp^*}=\delta^{-1}\langle 1,\lambda,\ldots,\lambda^{m-t-1} \rangle_{\F_q}$, where $\delta =f'(\lambda)$ and $f(x)$ is the minimal polynomial of $\lambda$ over $\fq$.
\end{proposition}

Recall that if $\lambda \in \F_{q^m}$, then the \textbf{degree} of $\lambda$ over $\F_q$ is the degree of the minimal
polynomial of $\lambda$ over $\F_q$, or equivalently, it is the smallest integer $e$ such that $\lambda \in \F_{q^e}$.
We extend the above lemma in the case that $\lambda$ has degree less than $m$ over $\F_q$.

\begin{proposition} \label{prop:duallambdasubfield}
    Assume that $m=se$, with $s,e>1$. Let $\lambda \in \F_{q^m}$ having degree $e$ over $\F_q$ and let $t\leq e$ be a positive integer. Then
    \[
(\langle 1,\lambda,\ldots,\lambda^{t-1}\rangle_{\F_q})^{\perp^*}=\mathrm{Ker}(\mathrm{Tr}_{q^m/q^e}) \oplus c \langle 1,\lambda,\ldots,\lambda^{e-t-1}\rangle_{\F_q},
    \]
    for some $c \in (\F_{q^m}^* \setminus \mathrm{Ker}(\mathrm{Tr}_{q^m/q^e}))$.
\end{proposition}

\begin{proof}
    Note that \begin{equation}  \label{eq:dualZ}
\F_{q^e}^{\perp^*_{q^m/q^e}}=\mathrm{Ker}(\mathrm{Tr}_{q^m/q^e})=:Z,
    \end{equation}
    and $Z$ is and $\F_{q^e}$-subspace of $\F_{q^m}$ having dimension $s-1$. %, and here $\perp^*_{q^m/q^e}$ is the orthogonal complement defined by $\mathrm{Tr}_{q^m/q^e}$. 
    %Hence, we have that \[Z=c'\langle \xi,\ldots,\xi^{s-1} \rangle_{\F_{q^e}},\]
    %for some $\xi \in \F_{q^m}$, such that $\F_{q^e}(\xi)=\F_{q^{m}}$ and $c' \in \F_{q^m}$, see e.g. \cite[Proposition 2.5]{napolitano2022clubs}. Note also that \[d':=\mathrm{Tr}_{q^m/q^e}(c') \neq 0,\]
    %otherwise, by \eqref{eq:dualZ}, $c' \in Z'$ and so $1 \in \langle \xi,\ldots,\xi^{s-1} \rangle_{\F_{q^e}}$, a contradiction.
    %Moreover, again 
    By \Cref{prop:dualwithdual}, we have that \begin{equation} \label{eq:dualbasissubfield}
(\langle 1,\ldots,\lambda^{t-1} \rangle_{\F_q})^{\perp^*_{{q^e/q}}}=d\langle 1,\ldots,\lambda^{e-t-1} \rangle_{\F_q},
    \end{equation}
    for some $d \in \F_{q^e}^*$, where here the orthogonal complement $\perp^*_{q^e/q}$ is defined with respect to $\mathrm{Tr}_{q^e/q}(\cdot)$. %Let
   % \[
    %c:=c'd'^{-1}d.
    %\]
    Let $c \in \F_{q^m}^*$ such that $\mathrm{Tr}_{q^m/q^e}(c)=d$. Note that $Z \cap  c \langle 1,\lambda,\ldots,\lambda^{e-t-1}\rangle_{\F_q}=\{0\}$. Indeed, assume that there exist $z \in Z$ and $u \in \langle 1,\lambda,\ldots,\lambda^{e-t-1}\rangle_{\F_q}$, such that $z=cu$. Therefore, since $Z$ is an $\F_{q^e}$-subspace, we get that $c \in Z$ and so, by \eqref{eq:dualZ}, we have 
    \[
    0 \neq d=\mathrm{Tr}_{q^m/q^e}(c)=0,
    \]
    a contradiction. 
    Now, we prove that, with this choice of $c$, we have 
    \[
(\langle 1,\lambda,\ldots,\lambda^{t-1}\rangle_{\F_q})^{\perp^*}%=c\left(\langle\xi,\ldots,\xi^{s-1} \rangle_{\F_{q^e}}\oplus \langle 1,\lambda,\ldots,\lambda^{e-t-1}\rangle_{\F_q}\right)=
=Z \oplus c \langle 1,\lambda,\ldots,\lambda^{e-t-1}\rangle_{\F_q}.
    \]
    For $z \in Z, u \in \langle 1,\lambda,\ldots,\lambda^{e-t-1}\rangle_{\F_q}$ and $w \in \langle 1,\lambda,\ldots,\lambda^{t-1}\rangle_{\F_q}$, we get 
    \[
    \begin{array}{rl}
    \mathrm{Tr}_{q^m/q}((z+cu)w) &= \mathrm{Tr}_{q^m/q} (zw)+\mathrm{Tr}_{q^m/q}(cuw) \\
    %& =\mathrm{Tr}_{q^m/q} (c'z(d'^{-1}dw))+\mathrm{Tr}_{q^m/q}(uc'd'^{-1}dw) \\
    &  =\mathrm{Tr}_{q^e/q} \left(\mathrm{Tr}_{q^m/q^e}(zw)\right)+\mathrm{Tr}_{q^e/q} \left(\mathrm{Tr}_{q^m/q^e}(cuw) \right) \\
    &  =\mathrm{Tr}_{q^e/q} \left(\mathrm{Tr}_{q^m/q^e}(zw)\right)+\mathrm{Tr}_{q^e/q} \left((du)w \right) \\
    & = 0 + 0,
   \end{array}
   \]
   where the last equality follows by \eqref{eq:dualZ} and  \eqref{eq:dualbasissubfield}, proving the assertion.
\end{proof}

Finally, we also show how the duality $\perp'$ acts on the direct product of $\F_q$-subspaces of $\F_{q^m}$, that will be useful in the next. 

\begin{proposition} \label{prop:dualcomplement}
    Let $U_1,\ldots,U_k$ be $\F_q$-subspaces of $\F_{q^m}$ and define $U=U_1 \times \ldots \times U_k \subseteq \F_{q^m}^k$. Then
    \[
    U^{\perp'}=U_1^{\perp^*} \times \ldots \times U_k^{\perp^*}.
    \]
\end{proposition}

\begin{proof}
For $(b_1,\ldots,b_k) \in U_1^{\perp^*} \times \ldots \times  U_k^{\perp^*}$, and for every $(a_1,\ldots,a_k) \in U$, 
\[ \begin{array}{rl}
   \sigma'((a_1,\ldots,a_k),(b_1,\ldots,b_k))  & =\mathrm{Tr}_{q^m/q}(a_1b_1+\ldots+a_kb_k)  \\
     & =\mathrm{Tr}_{q^m/q}(a_1b_1)+\ldots +\mathrm{Tr}_{q^m/q}(a_kb_k) \\
     & =0,
\end{array}
\]
implying that \begin{equation} \label{eq:inclusiondual}
U_1^{\perp^*} \times \ldots \times U_k^{\perp^*} \subseteq U^{\perp'}.\end{equation} 
Now, if $\dim_{\F_q}(U_i)=n_i$, for every $i$, we have that $\dim_{\F_q}(U)=n_1+\ldots+n_k$ and hence, $\dim_{\F_q}(U^{\perp'})=km-(n_1+\ldots+n_k)$. On the other hand, $\dim_{\F_q}(U_i^{\perp^*})=m-n_i$ and so $\dim_{\F_q}(U_1^{\perp^*} \times \cdots \times U_k^{\perp^*})=km-(n_1+\ldots+n_k)$. Therefore \[\dim_{\F_q}(U^{\perp'})=\dim_{\F_q}(U_1^{\perp^*} \times \cdots \times U_k^{\perp^*})\]
implying the equality in \eqref{eq:inclusiondual}.
\end{proof}

\subsection{Product of subspaces} \label{subsec:product}
The study of the properties of the product of subspaces in a finite field extension is related to the linear analogue of some problems arising from additive combinatorics, such as Kneser's Theorem \cite{kneser1956summenmengen} and Vosper's theorems \cite{lev2006critical,vosper1956critical}. In recent years, such investigations have been related to different contexts, such as in designing probabilistic decoding algorithms of LRPC codes \cite{gaborit2013low,franch2024bounded}, proposed in ROLLO \cite{melchor2020rollo} for the NIST post-quantum cryptography competition, and in the classifications of geometric objects \cite{napolitano2022clubs,napolitano2023classifications}. In the following, we recall some useful results and provide a characterization for pairs $(U_1,U_2)$ of $\F_q$-subspaces such that the $\F_q$-span of their product is an $\F_q$-hyperplane of $\F_{q^m}$.

For a subset $S$ of $\F_{q^m}$, we denote by $\langle S \rangle_{\F_{q}}$, the $\F_q$-subspace generated by $S$ and for two $\fq$-subspaces $U_1$ and $U_2$ of $\F_{q^m}$, we define \[U_1U_2:=\langle \{a_1 a_2 \colon a_1 \in U_1, a_2\in U_2\}\rangle_{\F_q}.\]

Note, that if $U_1=\langle a_1,\ldots,a_t \rangle_{\F_q}$ and $U_2$ is an $\F_q$-subspace of $\F_{q^m}$, we have that
\begin{equation} \label{eq:productsplitted}
U_1U_2=a_1U_2+\ldots+a_tU_2,
\end{equation}
where the sum on the right-hand side is intended to be the sum of $\F_q$-vector spaces. Moreover, via the duality $\perp^{*}$ defined in the above section, we also have 
\begin{equation}
    \label{eq:productsplittedperp}
(U_1U_2)^{\perp^*}=(a_1U_2+\ldots+a_tU_2)^{\perp^*}=a_1^{-1}U_2^{\perp^*}\cap \ldots \cap a_t^{-1}U_2^{\perp^*}.
\end{equation}

For prime extensions, the linear analog of the Cauchy-Davenport inequality reads as follows. 

\begin{theorem} [see \textnormal{\cite[Corollary 2.5]{hou2002generalization} and \cite[Theorem 3]{bachoc2018revisiting}}] \label{teo:bachocserrazemor}
Assume that $m$ is a prime. Let $U_1,U_2$ be $\F_q$-subspaces of $\F_{q^m}$ such that $\dim_{\F_q}(U_1),\dim_{\F_q}(U_2) \geq 1$ and $\dim_{\F_q}( U_1U_2 ) \leq m-1$. Then  \begin{equation} \label{eq:chauchydav}
    \dim_{\F_q}( U_1U_2 )\geq \dim_{\F_q}(U_1)+\dim_{\F_q}(U_2)-1.
  \end{equation}
\end{theorem}

The following classification of pairs of subspaces $(U_1,U_2)$ attaining the equality in \eqref{eq:chauchydav}, also known as \emph{critical pair}, holds.

\begin{theorem} [see \textnormal{\cite[Theorem 3]{bachoc2017analogue} and \cite[Lemma 5]{bachoc2017analogue}}] \label{teo:bachocserrazemorclass}
Assume that $m$ is a prime. Let $U_1,U_2$ be $\F_q$-subspaces of $\F_{q^m}$ such that $2 \leq \dim_{\F_q}(U_1),\dim_{\F_q}(U_2)$. 
\begin{enumerate}
    \item If $\dim_{\F_q}(U_1U_2) \leq m-2$ and 
\begin{equation}
\dim_{\F_q}(U_1U_2) =\dim_{\F_q}(U_1)+\dim_{\F_q}(U_2)-1,
\end{equation}
then $U_1=c_1 \langle 1,\lambda,\ldots,\lambda^{\dim_{\F_q}(U_1)-1}\rangle_{\F_q}$ and $U_2=c_2 \langle 1,\lambda,\ldots,\lambda^{\dim_{\F_q}(U_2)-1}\rangle_{\F_q}$, for some $c_1,c_2,\lambda \in \F_{q^m}^*$.
\item If $\dim_{\F_q}(U_1U_2) \leq m-1$, $U_1=c_1 \langle 1,\lambda,\ldots,\lambda^{\dim_{\F_q}(U_1)-1}\rangle_{\F_q}$, for some $c_1,\lambda \in \F_{q^m}$ and 
\begin{equation} \label{eq:eqproductm-1}
\dim_{\F_q}( U_1U_2) =\dim_{\F_q}(U_1)+\dim_{\F_q}(U_2)-1,
\end{equation}
then $U_2=c_2 \langle 1,\lambda,\ldots,\lambda^{\dim_{\F_q}(U_2)-1}\rangle_{\F_q}$, for some $c_2 \in \F_{q^m}^*$.
\end{enumerate}
\end{theorem}

In the next lemma, we characterize $\F_q$-subspaces of $\F_{q^m}$ satisfying \eqref{eq:eqproductm-1}, and for which $\dim_{\F_q}(U_1U_2)=m-1$.

\begin{lemma} \label{lem:gencriticalm-1}
    Let $U_1,U_2$ be $\F_q$-subspaces of $\F_{q^m}$, with $1 \leq \dim_{\F_q}(U_1),\dim_{\F_q}(U_2) <m$. Assume that $\dim_{\F_q}(U_2)=m-\dim_{\F_q}(U_1)$. If
    \[
    \dim_{\F_q}(U_1U_2)=m-1,
    \]
    then $U_2=cU_1^{\perp^*}$, for some $c \in \F_{q^m}^*$.
\end{lemma}

\begin{proof}
Since $\dim_{\F_q}(U_1U_2)=m-1$, there exists $b \in \F_{q^m}^*$, such that $U_1U_2=\mathrm{Ker}(\mathrm{Tr}_{q^m/q}(bx))$, see e.g. \cite{lidl1997finite}. This means that for any $w\in U_2$, we have that 
\[
\mathrm{Tr}_{q^m/q}(buw)=0,
\]
for every $u \in U_1$. Hence, $bw \in U_1^{\perp^*}$. This implies that $bU_2 \subseteq U_1^{\perp^*}$. Finally, by a dimensional argument, we get the assertion.
\end{proof}

\subsection{Rank-metric codes} \label{subsec:rank}

The \textbf{(rank) weight} $\w(\mathbf{v})$ of a vector $\mathbf{v}=(v_1,\ldots,v_n) \in \F_{q^m}^n$ is the $\F_q$-dimension of the vector space generated over $\F_q$ by its entries, i.e, $\w(\mathbf{v})=\dim_{\fq} (\langle v_1,\ldots, v_n\rangle_{\fq})$. The rank distance between two vectors $\uu, \vv \in \F_{q^m}^n$ is defined as $
d(\uu,\vv)=\w(\uu-\vv).$

An $[n,k,d]_{q^m/q}$ \textbf{(rank-metric) code} $\C $ is a $k$-dimensional $\F_{q^m}$-subspace of $\F_{q^m}^n$ endowed with the rank distance, where the parameter $d$ is its \textbf{minimum distance}, and it is defined as $
d=\min\{d(\cc_1,\cc_2) \colon \cc_1, \cc_2 \in \C, \cc_1 \neq \cc_2  \}.$ We also write that $\C$ is an  $[n,k]_{q^m/q}$ code if the parameter $d$ is not known/relevant. 
A \textbf{generator matrix} for $\C$ is a matrix $G\in \F_{q^m}^{k\times n}$, such that $\C=\{\xx G \colon \xx \in \F_{q^m}^k\}$. \\
Moreover, we denote by $A_i(\C)$, or simply by $A_i$, the number of codewords in $\C$ of weight $i \in \{0,\ldots,n\}$ and $(A_0,\ldots,A_n)$ is called the \textbf{weight distribution} of $\C$.

A Singleton-like bound for a rank-metric code holds. For an $[n,k,d]_{q^m/q}$ code, it holds \begin{equation}\label{eq:boundgen}
mk \leq \max\{m,n\}(\min\{n,m\}-d+1),\end{equation}
see \cite{delsarte1978bilinear}. An $[n,k,d]_{q^m/q}$ code is said to be \textbf{maximum rank distance code (MRD)} if its parameters attains the equality in \eqref{eq:boundgen}.

%In this setting, we define the equivalence of codes taking into account the $\F_{q^m}$-linear isometries of $\F_{q^m}^n$ endowed with the rank metric, see e.g. \cite{berger2003isometries}. More, precisely 

We also recall the notion of support of a vector. 
Let $\Gamma=(\gamma_1,\ldots,\gamma_m)$ be an ordered $\fq$-basis of $\F_{q^m}$. For any vector $\mathbf{v}=(v_1, \ldots ,v_n) \in \F_{q^m}^n$ define the matrix $\Gamma(\vv)\in \F_{q}^{n \times m}$, where
$$v_{i} = \sum_{j=1}^m \Gamma (\vv)_{ij}\gamma_j, \qquad \mbox{ for all } i \in \{1,\ldots,n\},$$
that is $\Gamma(\vv)$ is the matrix expansion of the vector $\vv$ with respect to the basis $\Gamma$ of $\F_{q^m}$. The \textbf{(rank) support} of $\vv$ is defined as the column span of $\Gamma(\vv)$, i.e. 
$\supp(\vv)=\mathrm{colsp}(\Gamma(\vv)) \subseteq \fq^n.$ The support of a vector does not depend on the choice of $\Gamma$ and we can talk about the support of a vector without mentioning $\Gamma$, see e.g. \cite[Proposition 2.1]{alfarano2022linear} and clearly, $\w(\vv)=\dim_{\F_q}(\supp(\vv))$. Finally, the support of a code $\C$ is the sum of the supports of the codewords defining a basis for $\C$, i.e.
\begin{equation} \label{eq:supportcode}
\supp(\C):=\supp(\cc_1)+\ldots+\supp(\cc_k),
\end{equation}
where $\{\cc_1,\ldots,\cc_k\}$ is a basis of $\C$ and the sum on the right-hand side is intended to be of vector spaces, see e.g. \cite[Proposition 2.4]{alfarano2022linear}. We also refer to \cite{alfarano2022linear,martinez2017relative}, for more details on the rank support. We will say that two codes $\C$ and $\C'$ are \textbf{equivalent} if and only if there exist a matrix $A \in \mathrm{GL}(n,q)$ such that
$\C'=\C\cdot A=\{\cc A : \cc \in \C\}$. We say that an $\nk$ code is \textbf{nondegenerate} if $\supp(\C)=\F_q^n$, or equivalently if the $\F_q$-column span of any generator matrix of $\C$ has dimension $n$. 

\subsection{The geometry of rank-metric codes}

The geometric counterpart of rank-metric codes are the systems. 
 \begin{definition}
An $[n,k,d]_{q^m/q}$ (or simply $[n,k]_{q^m/q}$) \textbf{system} $U$ is an $\F_q$-subspace of $\F_{q^m}^k$ of dimension $n$, such that
$ \langle U \rangle_{\F_{q^m}}=\F_{q^m}^k$ and
$$ d=n-\max\left\{\dim_{\F_q}(U\cap H) \mid H \textnormal{ is an $\F_{q^m}$-hyperplane of }\F_{q^m}^k\right\}.$$
\end{definition}

Two $[n,k,d]_{q^m/q}$ systems $U$ and $U'$ are \textbf{equivalent} if there exists an invertible matrix $B\in\GL(k,\F_{q^m})$ such that
$$ U'=U\cdot B=\{\xx B \colon \xx \in U\}.$$

Rank-metric codes and systems are related in the following way. Let $\C$ be an $[n,k]_{q^m/q}$ code and $G$ be an its generator matrix. Then the $\F_q$-subspace $U$ obtained as the $\F_q$-span of the columns of $G$ is called a \textbf{system associated with} $\C$. Viceversa, let $U$ be an $[n,k]_{q^m/q}$ system. Define $G$ as the matrix whose columns are an $\F_q$-basis of $U$ and let $\C$ be the code generated by $G$. $\C$ is called a \textbf{code associated with} $U$.

%Now, let $\C$ be an $\nkd$ code and let $G \in \mathbb{F}_{q^m}^{k\times n}$ be an its generator matrix and let $U$ be the $\mathbb{F}_q$-span of the columns of $G$, i.e. $U$ is a system associated with $\C$. 
%Define the map
%\begin{equation} \label{eq:phiGdef}
%\begin{array}{rccl}
%\psi_{G}:& \fq^{n}& \longrightarrow &U \\
%&\lambda & \longmapsto & \lambda G^\top,\end{array}
%\end{equation}
%which turns out to be an $\fq$-linear isomorphism. 
For a vector $\xx=(x_1,\ldots,x_k) \in \F_{q^m}^k$, by $\xx^\perp$ we denote the orthogonal complement of $\langle \xx\rangle_{\F_{q^m}}$ with respect to the standard scalar product in $\F_{q^m}^k$, i.e. \[\xx^\perp=\{(y_1,\ldots,y_k) \in \F_{q^m}^k \colon \sum_{i=1}^k x_iy_i=0\}.\] The following holds.

\begin{theorem} [see \textnormal{\cite{Randrianarisoa2020ageometric}}] \label{th:connection}
Let $\C$ be an $[n,k,d]_{q^m/q}$ code and let $G$ be a generator matrix.
Let $U \subseteq \F_{q^m}^k$ be the $\F_q$-span of the columns of $G$. %Then, for every $\xx \in \mathbb{F}_{q^m}^k$
%$$\psi_{G}^{-1}(U \cap \xx^\perp)=\supp(\xx G)^{\perp},$$
%where $\supp(\xx G)^{\perp}$ denotes the orthogonal complement of $\supp(\xx G)$ with respect to the standard scalar product in $\fq^n$.
%In particular, 
The weight of any codeword $\xx G \in \C$ is
\begin{equation}\label{eq:relweight}
w(\xx G) = n - \dim_{\fq}(U \cap \xx^{\perp}).\end{equation}
\end{theorem}

As a consequence, for an $\nkd$ code $\C$, we have that
\begin{equation} \label{eq:distancedesign}
d=n - \max\left\{ \dim_{\fq}(U \cap H)  \colon H\mbox{ is an } \F_{q^m}\mbox{-hyperplane of }\F_{q^m}^k  \right\}.
\end{equation}

Moreover, it can be proved that two codes are equivalent if and only if their associated systems are equivalent. 
The above argument allows to estabilish a one-to-one correspondence between equivalence classes of $[n,k,d]_{q^m/q}$ systems and equivalence classes of $[n,k,d]_{q^m/q}$ codes with respect to equivalence relation, see \cite{Randrianarisoa2020ageometric, alfarano2022linear}.

\begin{remark} \label{rk: changebasissystem}
    We note that if $G \in \F_{q^m}^{k\times n}$ is a generator matrix for the code $\C$, then any code $\C'$ equivalent to $\C$, via some $A \in \GL(n,q)$, has generator matrix $G'$ of the form $G'=BGA,$
for some $B\in \GL(k,q^m)$ and $A \in \GL(n,q)$. Moreover, if the codes $\C$ and $\C'$ are equivalent, their associated systems $U$ and $U'$ are related by $U'=U \cdot B=\{\xx B \colon \xx \in U\}$, for some $B \in \GL(k,q^m)$. Finally, let $U$ be a system associated with a code $\C$. Then, any code associated with $U$, due to the different choices of the basis of $U$, is of the form $\C \cdot A=\{\cc A \colon \cc \in \C\}$, for some $A \in \GL(n,q)$.
\end{remark}

%For an element $\xx \in \F_{q^m}^k$ and an $\F_q$-subspace $U$ of $\F_{q^m}^k$, we define 
%\[
%\w_U(\xx):=\dim_{\F_q}(\langle \xx \rangle_{\F_{q^m}} \cap U),
%\]
%as the \textbf{weight} of $\xx$ in $U$. 

\begin{proposition}
Let $\C$ be a nondegenerate $[n,k]_{q^m/q}$ code $\C$ having as generator matrix $G$ and associated system $U$, such that $U$ is the $\F_q$-column span of $G$. For every $\xx \in \F_{q^m}^k$, the following equivalences hold
\begin{equation} \label{eq:relationweightdual}
\w(\xx G)=i \Longleftrightarrow \dim_{\F_q}(U\cap \xx^{\perp})=n-i \Longleftrightarrow \dim_{\F_q}(U^{\perp'}\cap \langle \xx \rangle_{\F_{q^m}})=m-i,
\end{equation}
\end{proposition}
\begin{proof}
    The assertion immediately follows by Proposition \ref{prop:weightdual} together with \eqref{eq:relweight}.
\end{proof}

\section{Direct sum of 1-dimensional MRD codes} \label{sec:directsum}

In this section, we investigate \emph{completely decomposable codes}, i.e. codes that are, up to equivalence, the (external) direct sum of $1$-dimensional MRD codes.

Let $\C_1,\ldots,\C_t$ be codes such that $\C_i$ is an $[n_i,k_i]_{q^m/q}$ code. We define the \textbf{(external) direct sum} $\C_1\oplus \ldots\oplus\C_t$ of $\C_1,\ldots,\C_t$ as the $[n_1+\ldots+n_t,k_1+\ldots+k_t]_{q^m/q}$ code defined as
\begin{equation} \label{eq:defdecomposable}
\bigoplus_{i=1}^t\C_i=\C_1 \oplus \ldots \oplus \C_t:=\{(\cc_1,\ldots,\cc_t) \colon \cc_i \in \C_i\} \subseteq \F_{q^m}^{n_1} \times \cdots \times \F_{q^m}^{n_t}=\F_{q^m}^{n_1+\ldots+n_t}.  
\end{equation}

We investigate codes as in \eqref{eq:defdecomposable}, where the $\C_i$'s are 1-dimensional MRD codes, or equivalently, 1-dimensional nondegenerate codes.

\begin{lemma} \label{lem:MRDnondegen}
    Let $\C$ be an $[n,1]_{q^m/q}$ code. Then $\C$ is an MRD code if and only if $\C$ is nondegenerate, or in other words $\C=\langle \cc \rangle_{\F_{q^m}}$, with $\w(\cc)=n$. 
\end{lemma}
 
\begin{proof}
Assume that $\C$ is an MRD code. Then $m=\max\{m,n\}(\min\{m,n\}-d+1)$, implying that 
\[
d=\min\{m,n\}-\frac{m}{\max\{m,n\}}+1.
\]
The minimum distance $d$ is a positive integer and so $\max\{m,n\}=m$, from which $n \leq m$ and $d=n$. Thus, $\C=\langle \cc \rangle_{\F_{q^m}}$, for some $\cc \in \F_{q^m}^n$ of weight $n$, implying that $\C$ is a nondegenerate code.
    Conversely, if $\C$ is nondegenerate then necessarily $n \leq m$ and $d=n$. Therefore, the equality in \eqref{eq:boundgen} holds. 
\end{proof}

In the following, for $\ell$ matrices $A_1,\ldots,A_{\ell}$, where $A_i \in \F_q^{t_i \times n_i}$, the direct sum $A_1\oplus \ldots \oplus A_{\ell}$ is the block matrix 
\[
A_1\oplus \ldots \oplus A_{\ell}=\begin{pmatrix}
            A_1 & 0 & 0 & \cdots & 0 \\
            0 & A_2 & 0 & \cdots  & 0 \\
            0 & 0 & \ddots & \cdots & 0 \\
            \vdots & \vdots & \vdots & \ddots & \vdots \\
            0 & 0 & 0 & \cdots & A_{\ell}
        \end{pmatrix} \in \F_q^{(t_1+\ldots+t_{\ell}) \times (n_1 +\ldots+n_{\ell})}.
\]

In the next result, we characterize codes that are, up to equivalence, the direct sum of 1-dimensional MRD codes as those that admit a basis ${\mathbf{c}_1, \ldots, \mathbf{c}_k}$ such that the sum of the weights of the $\mathbf{c}_i$'s equals the length of the code. Additionally, equivalent conditions in terms of support of the $\cc_i$'s and of generator matrix for such codes are provided.

\begin{theorem} \label{th:characterization1dim}
    Let $\C$ be an $\nk$ code. Let $m>n_1 \geq \ldots \geq n_k \geq 1$ be positive integers. Then the following are equivalent:
    \begin{enumerate}[$1)$]
        \item $\C$ is equivalent to a code $\C'=\bigoplus_{i=1}^k \C_i$, where $\C_i \mbox{ is an MRD } [n_i,1]_{q^m/q} \mbox{ code, for every }i$;
 \item $\C$ admits a basis $\{\cc_1,\ldots,\cc_k\}$ such that $
n=\sum\limits_{i=1}^kn_i$ and 
\begin{equation} \label{eq:directsumsupport}
\supp(\C)=\bigoplus\limits_{i=1}^k \supp(\cc_i),\end{equation}
with $\dim_{\F_q}(\supp(\cc_i))=n_i$, for every $i$;
\item $\C$ admits a basis $\{\cc_1,\ldots,\cc_k\}$ such that
$
n=\sum\limits_{i=1}^kn_i, \mbox{ with }n_i=\w(\cc_i)$, for every $i$;
        \item a generator matrix of a code equivalent to $\C$ is of the form 
        \[ G=\begin{pmatrix}
            \uu_1 & 0 & 0 & \cdots & 0 \\
            0 & \uu_2 & 0 & \cdots  & 0 \\
            0 & 0 & \ddots & \cdots & 0 \\
            \vdots & \vdots & \vdots & \ddots & \vdots \\
            0 & 0 & 0 & \cdots & \uu_k
        \end{pmatrix}=\mathbf{u}_1 \oplus \mathbf{u}_2 \oplus \ldots \oplus \mathbf{u}_k,\]
    where $\mathbf{u}_i \in \F_{q^m}^{1 \times n_i}$ and $\w(\uu_i)=n_i$, for each $i$.
    \end{enumerate}
\end{theorem}

\begin{proof} 
\emph{$\underline{1) \Rightarrow 2)}$.} Assume that $\C$ is equivalent to a code $\C'=\bigoplus_{i=1}^k \C_i$, where $\C_i$ is an MRD $[n_i,1]_{q^m/q}$ code, for every $i$. So, $\C_i=\langle \cc_i'\rangle_{\F_{q^m}}$, with $\w(\cc_i')=n_i$, for any $i$. Clearly $n=n_1+\cdots+n_k$. Now, since $\C$ is equivalent to $\C'$, there exists a matrix $A \in \GL(n,q)$ such that $\C=\C' \cdot A$. Then the codewords \[
\cc_1''=(\cc_1',\mathbf{0},\ldots,\mathbf{0})A, \cc_2''=(\mathbf{0},\cc_2',\mathbf{0},\ldots,\mathbf{0})A, \ldots,\cc_k''=(\mathbf{0},\ldots,\mathbf{0},\cc_k')A \in \C' \cdot A=\C,
\]
form a basis of $\C$. Note that 
\[
\supp(\cc_i'')= \supp((\mathbf{0},\cc_i',\mathbf{0},\ldots,\mathbf{0}))A = (\underbrace{(0,\ldots,0)}_{n_1+\cdots+n_{i-1}} \times \F_q^{n_i}\times \underbrace{(0,\ldots,0)}_{n_{i+1}+\cdots+n_{k}}  )A, 
\]
see e.g. \cite[Proposition 2.1]{alfarano2022linear}, from which we get the assertion.
\\

\emph{$\underline{2) \Rightarrow 3)}$.}  The assertion immediately follows by the fact that $\w(\cc)=\dim_{\F_q}(\supp(\cc))$, for any codeword $\cc \in \C$.
 
\emph{$\underline{3) \Rightarrow 4)}$.} Assume that there exists a basis $\{\cc_1,\ldots,\cc_k\}$ of $\C$ as in the hypotheses. Let $G'$ be a generator matrix for $\C$ having as rows the codewords $\cc_1,\ldots,\cc_k$. Therefore, $\cc_1=\ee_1G',\ldots,\cc_k=\ee_kG',$
where the $\ee_i$'s are the vectors of the standard basis of $\F_{q^m}^k$. Let $U$ be a system associated with $\C$ defined as the $\F_q$-column span of $G'$.
By using \eqref{eq:relweight}, we have that 
\[
\dim_{\F_q}(U \cap \ee_i^{\perp})=n-\w(\cc_i)    \mbox{ and } \sum_{i=1}^k \dim_{\fq}(U \cap \ee_i^{\perp})=kn-n. 
\]
Let consider the dual subspace $U^{\perp'} \subseteq \F_{q^m}^k$ of $U$. Then $\dim_{\F_q}(U^{\perp'})=km-n$ and by \eqref{eq:relationweightdual}, we get
\[
\sum_{i=1}^k \dim_{\F_q}(U^{\perp'} \cap \langle \ee_i \rangle_{\fqm})=km-n. 
\]
Hence, 
\[
U^{\perp'}=(U^{\perp'}\cap \langle \ee_1 \rangle_{\fqm}) \oplus \ldots \oplus (U^{\perp'} \cap \langle \ee_k \rangle_{\fqm}),
\]
i.e. $U^{\perp'}\cap \langle \ee_i \rangle_{\fqm}=\{(0,\ldots,0,w_i,0,\ldots,0) \colon w_i \in W_i\}$, where $W_i$ is an $\F_q$-subspace of $\F_{q^m}$ having dimension $m-n_i$. This implies that 
\[U^{\perp'}=W_1 \times W_2 \times \ldots \times W_k.\]
Now, by \Cref{prop:dualcomplement}, we have that \[U=(U^{\perp'})^{\perp'}=
(W_1\times \ldots \times W_k)^{\perp'}=U_1 \times U_2 \times \ldots \times U_k,
\]
with $U_i=W_i^{\perp^*}$, where $W_i^{\perp^*}=\{ a \in \F_{q^m} \colon \mathrm{Tr}_{q^m/q}(ab)=0, \mbox{ for any }b \in W_i \}$. Also, $\dim_{\F_q}(U_i)=m-\dim_{\F_q}(W_i)=n_i$. The assertion follows by the definition of system associated to $\C$ and Remark \ref{rk: changebasissystem}. 

\emph{$\underline{4) \Rightarrow 1)}$.} Clearly, the code having $G$ as a generator matrix is of the form $\bigoplus_{i}^k \langle \uu_i\rangle_{\F_{q^m}}$. Hence, $\C$ is equivalent to $\bigoplus_{i}^k \C_i$, where $\C_i=\langle \uu_i\rangle_{\F_{q^m}}$ and so the assertion follows.
\end{proof}

\begin{definition}
An $\nk$ code $\C$ is said to be \textbf{completely decomposable of type $(n_1,\ldots,n_k)$}, for some integers $m>n_1 \geq n_2 \geq \ldots \geq n_k \geq 1$ if one of the equivalent conditions in Theorem \ref{th:characterization1dim} occurs. A generator matrix $G$ of the form
\[
G=\uu_1 \oplus \ldots \oplus \uu_k \in \F_{q^m}^{k \times n},
\]
for some $\mathbf{u}_i \in \F_{q^m}^{1 \times n_i}$ such that $\w(\uu_i)=n_i$, for each $i$, will be said to be in \textbf{a weight complementary form}.
\end{definition}

By Theorem \ref{th:characterization1dim}, we immediately obtain the following.

\begin{proposition}
   A completely decomposable $\nk$ code $\C$ of type $(n_1,\ldots,n_k)$ is nondegenerate.
\end{proposition}

\begin{remark} \label{rk:systemassociatedtype}
    Note that by \Cref{th:characterization1dim}, we also deduce that there exists a system $U$ associated with a completely decomposable code $\C$ of type $(n_1,\ldots,n_k)$ of the form 
        \[
        U=U_1 \times U_2 \times \ldots \times U_k \subseteq \F_{q^m}^k, 
        \]
        where $U_i$ is an $\F_q$-subspace of $\F_{q^m}$ having dimension $n_i$, for each $i$. Subspaces of this shape have been recently studied in connection with linear sets having points of complementary weights and multi-orbits cyclic codes; see \cite{adriaensen2023minimum,napolitano2022linearsets,zullo2023multi}.
\end{remark}

\begin{remark}
\Cref{th:characterization1dim} shows that codes that are, up to equivalence, direct sum of $1$-dimensional MRD codes can be also characterized as those nondegenerate codes admitting a basis $\{\cc_1,\ldots,\cc_k\}$ such that the supports of $\cc_1,\ldots,\cc_k$ are in direct sum. So, for such kind of codes the left-hand side of \eqref{eq:supportcode} is a direct sum of subspaces. 
\end{remark}

In the next, we prove that the type of a completely decomposable code is uniquely determined.

\begin{proposition}
    Let $\C$ be a completely decomposable $\nk$ code of type $(n_1,\ldots,n_k)$. If $\C$ is equivalent to a completely decomposable code $\C'$ of type $(n'_1,\ldots,n'_k)$, then $n_i=n'_i$, for any $i$.
\end{proposition}

\begin{proof}
Assume that $\C$ is equivalent to a completely decomposable code $\C'$ of type $(n'_1,\ldots,n'_k)$ and suppose, by contradiction that 
$n_1=n_1', \ldots, n_{\ell}=n_{\ell}'$ and  $n_{\ell+1}\neq n_{\ell+1}$, for some $\ell \in \{0,\ldots,k-1\}$. Without loss of generality, suppose that $n_{\ell+1}>n'_{\ell+1}$.
By \Cref{rk:systemassociatedtype}, we know that, up to equivalence, there exist a system of the form $U=U_1 \times \ldots \times U_k$ and a system of the form $U'=U_1' \times \ldots \times U_k'$ associated with $\C$ and with $\C'$, respectively, where $\dim_{\F_q}(U_i)=n_i$ and $\dim_{\F_q}(U'_i)=n'_i$.  Since $\C$ is equivalent to $\C'$, then the associated systems are equivalent. Therefore, there exists a matrix $B \in \GL(k,q^m)$, such that
\[
U'=U\cdot B=\{\xx B \colon \xx \in U\},
\]
see Remark \ref{rk: changebasissystem}.
Let consider $W=\langle \ee_1,\ldots,\ee_{\ell+1}\rangle_{\F_{q^m}},$ where the $\ee_i$'s are the vectors of the standard basis of $\F_{q^m}^k$, and let $W'=W \cdot B=\{\uu B \colon \uu \in W\}$. Then 
\begin{equation} \label{eq:partialequalitylenght}
\begin{array}{rl}

\dim_{\F_q}(U' \cap W') &  = \dim_{\F_q}((U \cdot B) \cap (W \cdot B)) \\
&  = \dim_{\F_q}((U \cap W) \cdot B) \\
&= \dim_{\F_q}(U \cap W ) \\
&= \sum_{i=1}^{\ell+1} \dim_{\F_q}(U_i)=n_1+\ldots +n_{\ell+1} \\
&=n_1'+\cdots+n'_{\ell}+n_{\ell+1}.
\end{array}
\end{equation}
We complete a basis of $W'$ to a basis of $\F_{q^m}^k$ by adding some vectors $\ee_{i_{\ell+2}},\ldots,\ee_{i_{k}}$ of the standard basis $\{\ee_1,\ldots,\ee_k\}$.
Therefore 
\[
\begin{array}{rl}
n&=n'_1+\ldots+n'_k \\
&=\dim_{\F_q}(U')\\
&=\dim_{\F_q}(U' \cap (W ' \oplus \langle \ee_{i_{\ell+2}}\rangle_{\F_{q^m}} \ldots \oplus  \langle \ee_{i_{k}}\rangle_{\F_{q^m}}))) \\
& \geq \dim_{\F_q}(U' \cap W ' )+ \dim_{\F_q}(U' \cap \langle \ee_{i_{\ell+2}}\rangle_{\F_{q^m}} )+\ldots+\dim_{\F_q}(U' \cap \langle \ee_{i_{k}}\rangle_{\F_{q^m}})) \\
&=n'_1+\ldots+n_{\ell}'+n_{\ell+1}+n'_{i_{\ell+2}}+\ldots+n'_{i_{k}} \hskip 2.9 cm \mbox{(by \eqref{eq:partialequalitylenght})} \\
& \geq n'_1+\ldots+n_{\ell}'+n_{\ell+1}+n'_{\ell+2}+\ldots+n'_{k}  \hskip 3 cm \mbox{(by $n'_1 \geq \cdots \geq n'_k$ )} \\
& >n'_1+\ldots+n_{\ell}'+n'_{\ell+1}+n'_{\ell+2}+\ldots+n'_{k} \hskip 3 cm \mbox{(by $n_{\ell+1}>n'_{\ell+1}$ )}\\
& =n,
\end{array}
\]
a contradiction.
\end{proof}

Completely decomposable codes are examples of rank-metric codes that are not minimal but for which the set of minimal codewords can be described. The notion of minimal code in the rank metric was introduced and studied in \cite{alfarano2022linear},
where an its geometric characterization was provided.

Let $\C$ be an $[n,k]_{q^m/q}$ code. A nonzero codeword $\cc \in \C$ is a \textbf{minimal codeword} if for every nonzero $\cc' \in \C$, $\supp(\cc') \subseteq \supp(\cc)$ implies that $\cc=\alpha \cc'$, for some $\alpha \in \F_{q^m}^*$. We say that $\C$ is
\textbf{minimal} if all its nonzero codewords are minimal. 

\begin{proposition} \label{prop:descriptionminimal}
    Let $\C$ be a completely decomposable $\nk$ code of type $(n_1,\ldots,n_k)$, with $k \geq 2$ and $G=\uu_1\oplus \ldots \oplus \uu_k$ be an its generator matrix in a weight complementary form. Then $\C$ is not a minimal code and the set of minimal codewords of $\C$ is
    \[
    \bigcup_{i=1}^k \left\{(\mathbf{0},\ldots,\mathbf{0},\alpha \uu_i,\mathbf{0},\ldots,\mathbf{0}) \colon \alpha \in \F_{q^m}^*\right\}.
    \]
\end{proposition}

\begin{proof}
    Let $\cc$ be a nonzero codeword of $\C$. Then $\cc$ is of the form $\cc=(\alpha_1 \uu_1,\ldots,\alpha_k \uu_k)$, for some $\alpha_i \in \F_{q^m}$. It is easy to see that if there exist two different indices $i\neq j$ such that $\alpha_i,\alpha_j \neq 0$, then the codeword $\cc'=(\mathbf{0},\ldots,\mathbf{0},\alpha_i \uu_i,\mathbf{0},\ldots,\mathbf{0})$ is such that $\supp(\cc') \subseteq \supp(\cc)$. However, $\cc$ and $\cc'$ are not $\F_{q^m}$-proportional and so $\cc$ is not minimal. Therefore, if $\cc$ is a minimal codeword of $\C$, then $\cc=(\mathbf{0},\ldots,\mathbf{0},\alpha_i \uu_i,\mathbf{0},\ldots,\mathbf{0})$ for some $i$ and $\alpha_i \neq 0$. Now, we prove that if 
    $\cc=(\mathbf{0},\ldots,\mathbf{0},\alpha_i \uu_i,\mathbf{0},\ldots,\mathbf{0})$ then it is minimal. Indeed,
    note that  
    \[\supp(\cc)= \supp((\mathbf{0},\ldots,\mathbf{0},\alpha_i \uu_i,\mathbf{0},\ldots,\mathbf{0}))\subseteq 
    \{0\}^{n_1+\cdots+n_{i-1}} \times \F_q^{n_i}\times \{0\}^{n_{i+1}+\cdots+n_{k}}  .\] 
    By using that $\dim_{\F_q}(\supp(\cc))=\w(\cc)=n_i$, we get that 
    \[\supp(\cc)=
    \{0\}^{n_1+\cdots+n_{i-1}} \times \F_q^{n_i}\times \{0\}^{n_{i+1}+\cdots+n_{k}}  .\] 
    Let $\cc'=(\beta_1 \uu_1,\ldots,\beta_k \uu_k)$ be a nonzero codeword of $\C$ such that $\supp(\cc') \subseteq \supp(\cc)$. Observe that, for every $j \in \{1,\ldots,k\}$, we have that $\supp(\cc_j) \subseteq \supp(\cc')$, where $\cc_j=(\mathbf{0},\ldots,\mathbf{0},\beta_j\uu_j,\mathbf{0},\ldots,\mathbf{0})$. Therefore, using that $\supp(\cc') \subseteq \supp(\cc)$, we get $\beta_j=0$, for any $j \neq i$, implying that $\cc=\alpha \cc'$, for some $\alpha \in \F_{q^m}$. This proves the assertion. 
\end{proof}

\begin{comment}
\begin{proposition} [see \textnormal{\cite[Proposition 6.1.]{zullo2023multi}}]
    Let $U$ be an $\F_q$-subspace of $V(k,q^m)$ with $\dim_{\F_q}(U)=n$. Assume that there exists a basis $\{ \mathbf{w}_1,\ldots,\mathbf{w}_k\}$ of $\F_{q^m}^k$ such that 
    \[
    \sum_{i=1}^{k} \dim_{\fq}(U \cap \langle \mathbf{w}_i \rangle_{\fqm})=n.
    \]
    Then $W$ is $\mathrm{GL}(k,q^m)$-equivalent to the $\F_q$-subspace 
    \[
    U=U_1\times U_2 \times \ldots \times U_k,
    \]
    where $U_i$ is an $\F_q$-subspace of $\F_{q^m}$ having dimension $\dim_{\fq}(U \cap \langle \mathbf{w}_i \rangle_{\fqm})$, for each $i$.
\end{proposition}
\end{comment}

Finally, we investigate the dual operation on completely decomposable codes. Recall that, the \textbf{dual code} of a code $\C$ is
\[ \C^\perp = \left\{ (u_1,\ldots,u_n) \in \F_{q^m}^n \colon \sum_{i=1}^n u_iv_i=0, \mbox{ for every }(v_1,\ldots,v_n) \in \C \right\}. \]
If $\C=\bigoplus_{i=1}^k\C_i$ is a completely decomposable $[n,k]_{q^m/q}$ code of type $(n_1,\ldots,n_k)$ then its dual is $\C^{\perp}=\bigoplus_{i=1}^k\C_i^{\perp_i}$, where $\C_i^{\perp_i}$ is the dual of $\C_i$ in $\F_{q^m}^{n_i}$, for every $i$. In general, it is not the direct sum of $1$-dimensional MRD codes. This means the property of being completely decomposable is not preserved under duality. On the other hand, we prove that such property is preserved under the operation of the geometric dual, recently introduced in \cite{borello2023geometric}. 

\begin{definition}
Let $\C$ be a nondegenerate $[n,k,d]_{q^m/q}$ and let $U$ be a system associated with $\C$. Assume that $\dim_{\F_q}(U\cap \langle v\rangle_{\F_{q^m}})<m$, for any $v \in \F_{q^m}^k$.
Then a \textbf{geometric dual} of $\C$ (with respect to $\perp'$) is defined as $\C'$, where $\C'$ is any code associated with the system $U^{\perp'}$.
\end{definition}

\begin{remark}
In order to give the definition of the geometric dual of a rank-metric code, we have assumed that $\dim_{\F_q}(U \cap \langle v \rangle_{\F_{q^m}}) < m$ for any $v \in \F_{q^m}^k$. This condition is equivalent to the one used originally in \cite{borello2023geometric}, which relies on the notion of generalized rank weights. %The generalized weights have been introduced several times with different definitions; see, e.g., \cite{jurrius2017defining}. The aforementioned equivalence immediately follows from the definition of generalized weights given in \cite{Randrianarisoa2020ageometric}, and more precisely, from the equivalent one given in \cite[Theorem 3.14]{alfarano2022linear}, which is directly connected with the system associated with the code.
\end{remark}

\begin{proposition}
    Let $\C$ be a completely decomposable $\nk$ code of type $(n_1,\ldots,n_k)$. Let $U$ be a system associated with $\C$. Then $\dim_{\F_q}(U\cap \langle v\rangle_{\F_{q^m}})<m$, for any $v \in \F_{q^m}^k$ and a geometric dual $\C'$ of $\C$ is a completely decomposable $[km-n,k]_{q^m/q}$ code of type $(m-n_k,\ldots,m-n_1)$.
\end{proposition}

\begin{proof}
By \Cref{rk:systemassociatedtype}, we know that, up to equivalence, there exists a system associated with $\C$ of the form 
        \[
        U=U_1 \times U_2 \times \ldots \times U_k \subseteq \F_{q^m}^k. 
        \]    
First, we prove that $\dim_{\F_q}(U\cap \langle v\rangle_{\F_{q^m}})<m$, for any $v \in \F_{q^m}^k$. Clearly, the assertion is true, if $m > n$. So, assume that $m\leq n$ and suppose by contradiction that there exists a nonzero $\uu \in \F_{q^m}^k$, such that $\dim_{\fq}(U \cap \langle \uu \rangle_{\F_{q^m}}) =m$. Then, there exists some $j \in \{1,\ldots,k\}$, such that \[\{\ee_i \colon i=1,\ldots,k \mbox{ and } i \neq j\} \cup \{\uu\}\] is a basis of $\F_{q^m}^k$. Therefore,
\[
n=n_1+\ldots+n_k=\dim_{\F_q}(U) \geq \sum_{i \neq j}\dim_{\F_q}(U \cap \langle \ee_i \rangle_{\F_{q^m}}) + \dim_{\F_q}(U \cap \langle \uu \rangle_{\F_{q^m}})=n-n_j+m>n,
\]
a contradiction. So, $\dim_{\F_q}(U\cap \langle v\rangle_{\F_{q^m}})<m$, for any $v \in \F_{q^m}^k$.
Now, by definition, a system associated with $\C'$ is 
\[
U^{\perp'}=(U_1\times \ldots \times U_k)^{\perp'}=
U_1^{\perp^*} \times \ldots \times U_k^{\perp^*},\]
with $\dim_{\F_q}(U_i)=m-n_i$, where the last equality follows from \Cref{prop:dualcomplement}.
Then the assertion follows by \Cref{th:characterization1dim}. 
\end{proof}

\section{On the weight distribution of completely decomposable codes}

In general, determining the weight distribution of a code from its structure is a challenging task, even providing some information about it is difficult. For rank-metric codes, the weight distribution of optimal codes is uniquely determined, see \cite{delsarte1978bilinear}. Except for such codes, very few is known about the weight distribution for other classes of codes. For some classes of $k$-dimensional $\F_{q^m}$-linear rank metric codes in $\F_{q^m}^n$, particularly when $k \in \{2,3\}$, the weight distributions have been determined relying on their specific geometric structures, see, e.g., \cite{napolitano2022clubs, polverino2023maximum, lia2023short, polverino2022divisible, marino2023evasive}.

In this section, we show that for completely decomposable codes, we can describe some properties on the weight distribution and codewords having certain weights. Also, we provide the number of minimum weight codewords of a completely decomposable code relying on properties of $\F_q$-subspaces of $\F_{q^m}$ associated with the code. As a consequence, we determine bounds on the number of minimum weight codes and we also provide constructions attaining such bounds. Finally, we characterize completely decomposable codes having the maximum number of minimum weight codewords.

Let $\C$ be a completely decomposable codes of type $(n_1,\ldots,n_k)$. By definition, up to equivalence, a generator matrix of $\C$ is of the form $\mathbf{u}_1 \oplus \mathbf{u}_2 \oplus \ldots \oplus \mathbf{u}_k,$ where $\mathbf{u}_i \in \F_{q^m}^{1 \times n_i}$ and $\w(\uu_i)=n_i$, for each $i$. In other words,
\begin{equation} \label{eq:equivdirect}
    \C=\bigoplus_{i=1}^k \C_i,
    \end{equation}
with $\C_i =\langle \uu_i \rangle_{\F_{q^m}}$, for every $i$. Since the metric properties of a code are invariant under equivalence, we often consider $\C$ as in \eqref{eq:equivdirect}.

%Let $I:=\{i_1,\ldots,i_k\}\subseteq \{1,\ldots,n\}$. Let $\pi_{I}:(u_1,\ldots,u_k) \in \F_{q^m}^n \mapsto (u_{i_1},\ldots,u_{i_t}) \in \F_{q^m}^{\lvert I \rvert}$ be the projection on the coordinates indexed by $I$. For an $[n,k]_{q^m/q}$ code $\C$, the \textbf{$I$-shortened code} of $\C$ is defined as the code \[\C^I \subseteq \F_{q^m}^k.\] Let $\C$ be a completely decomposable $[n,k]_{q^m/q}$ code of type $(n_1,\ldots,n_k)$, if $I=\{n_1+\ldots+n_{t-1}+1,\ldots,n_1+\cdots+n_k\}$, we denote by $\pi_t:=\pi_{I}$, the projection on the last $n_t+\cdots+n_k$ components. Note that, for any $t \in \{1,\ldots,k\}$, $\pi_t(\C)$ is a completely decomposable $[n_t+\cdots+n_k,k-t+1,n_k]$ code of type $(n_t,\ldots,n_k)$, and we simply call $\pi_t(\C)$ the \textbf{$t$-punctured code of $\C$}. 

We start with the following lower bound on the weight of a codeword $\cc$ of a completely decomposable code, depending on how we write $\cc$ as a linear combination of the fixed basis of $\C$. 

\begin{proposition} \label{prop:boundweightcompl}
Let $\C$ be a completely decomposable $\nk$ code of type $(n_1,\ldots,n_k)$ and $G=\uu_1\oplus \ldots \oplus \uu_k$ be an its generator matrix in a weight complementary form. Then any nonzero codeword $\cc=(\alpha_1 \uu_1,\ldots,\alpha_k \uu_k)$ has weight
   \begin{equation} \label{eq:boundweightcompl}
   \w(\cc) \geq \max\{n_j : \alpha_j \neq 0\}.
   \end{equation}
   In particular, $\C$ is an $[n,k,n_k]_{q^m/q}$ code. 
\end{proposition}

\begin{proof}
    The assertion immediately follows by the fact that if a codeword $\cc=(\alpha_1 \uu_1,\ldots,\alpha_k \uu_k) \in \C$ is such that $\alpha_j \neq 0$, for some $j \in \{1,\ldots,k\}$, then 
    \[
    \w(\cc) \geq \w(\uu_j)=n_j.
    \]
\end{proof}

Proposition \ref{prop:boundweightcompl} and the Singleton-like bound imply that for any completely decomposable $\nk$ code $\C$ of type $(n_1,\ldots,n_k)$, it holds
\begin{equation} \label{eq:singletoncompletely}
mk\leq \max\{m,n\}(\min\{m,n\}-n_k+1).
\end{equation}

Clearly if $n_1=\cdots=n_k=1$, with $k>1$, or in other words $n=k$, then $\C=\F_{q^m}^k$ and it is an MRD code. Now, we show that completely decomposable codes are not MRD codes for $n>k$. 

\begin{proposition}
    A completely decomposable $[n,k]_{q^m/q}$ code $\C$ of type $(n_1,\ldots,n_k)$, with $n>k>1$, is not an MRD code.
\end{proposition}

\begin{proof}
First assume that $m \geq n$. Then, \eqref{eq:singletoncompletely} reads as $k=n-n_k+1=(n_1+\cdots+n_{k-1}+1)$ and so $n_1=\ldots=n_k=1$ and $n=k$, a contradiction to our assumptions. Now, assume that $n>m$. We known that, up to equivalence, there exists a system associated with $\C$ of the form $U=U_1 \times \ldots \times U_k$, where $U_i$ is an $\F_q$-subspace of $\F_{q^m}$ having dimension $n_i$. If $\C$ is an MRD codes, then by \cite[Section 5]{zini2021scattered}, see also \cite[Theorem 4.9]{marino2023evasive}, we get that $n=mk/(h+1)$, for some positive integer $h$ such that $h+1 \mid mk$ and $\dim_{\F_q}(U \cap W) \leq h$, for any $\F_{q^m}$-subspace $W$ of $\F_{q^m}^k$ having dimension $h$. Therefore, \cite[Proposition 2.1]{csajbok2021generalising} implies that
$\dim_{\F_q}(U \cap \langle \xx \rangle_{\F_{q^m}}) \leq 1$, for any $\xx \in \F_{q^m}^k$. Since $\dim_{\F_q}(U \cap \langle \mathbf{e}_i \rangle_{\F_{q^m}}) = n_i$, where the $\mathbf{e}_i$'s are standard basis vectors, we have that $n_1=\ldots=n_k=1$. So, $n=k$ and we still get a contradiction.  
\end{proof}

In the next, we show some properties on the weight distribution of completely decomposable codes. 
For this aim, the operations of shortening and puncturing on rank-metric codes are very useful. %Shortening is a fundamental coding theoretic operation and arises in numerous contexts. 
We refer to \cite{byrne2017covering,martinez2016similarities} for their properties.

We will use special type of shortening and puncturing based on the structure of completely decomposable codes.

\begin{definition}
 Let $\C$ be a completely decomposable $[n,k]_{q^m/q}$ code of type $(n_1,\ldots,n_k)$ and $G=\uu_1\oplus \ldots \oplus \uu_k$ be an its generator matrix in a weight complementary form. Let $t \in \{1,\ldots,k\}$. We call
\[
\Sigma(\C,t):=\{(\alpha_1\uu_1,\ldots,\alpha_k\uu_k) \in \C \mid \alpha_i=0, \mbox{ for any }i=1,\ldots,t-1\} \subseteq \F_{q^m}^n
\]
and 
\[
\Pi(\C,t):=\{(\alpha_t\uu_t,\ldots,\alpha_k\uu_k) \mid \alpha_i \in \F_{q^m} \} \subseteq \F_{q^m}^{n_{t}+\cdots+n_k}
\]
the \textbf{$t$-shortened code} and \textbf{$t$-punctured code} of $\C$, respectively.
\end{definition}
Note that, for any $t \in \{1,\ldots,k\}$, $\Sigma(\C,t)$ is an $[n,k-t+1,n_k]_{q^m/q}$ code and $\Pi(\C,t)$ is a completely decomposable $[n_t+\cdots+n_k,k-t+1,n_k]_{q^m/q}$ code of type $(n_t,\ldots,n_k)$.

By Proposition \ref{prop:boundweightcompl}, if a codeword $\cc \notin \Sigma(\C,t+1)$, we have that $\w(\cc)\geq n_t$. In the next result, we describe the codewords contained in $\Sigma(\C,t) \setminus \Sigma(\C,t+1)$ having weight $n_t$.

\begin{lemma}
\label{lm:numberpointsweightcomplementary}
Let $\C$ be a completely decomposable $\nk$ code of type $(n_1,\ldots,n_k)$. Let $G=\uu_1\oplus \ldots \oplus \uu_k$ be an its generator matrix in a weight complementary form and let $U_i$ be the $\F_q$-span of the entries of $\uu_i$, for every $i$.
For any integer $t \in \{1,\ldots,k-1\}$, the set of codewords of $\Sigma(\C,t) \setminus \Sigma(\C,t+1)$ of weight $n_t$ is
\[ 
\mathcal{W}_t=\left\{
\begin{array}{rl}
\beta(\mathbf{0},\ldots,\mathbf{0},\uu_t,\xi_{t+1}\uu_{t+1}, \ldots, \xi_{k}\uu_k) \colon & 
     \xi_h \in   \left( U_t^{\perp^*} U_h \right)^{\perp^*},  
     \\ & \mbox{for } h \in \{t+1,\ldots,k\} \mbox{ and }\beta \in \F_{q^m}^*
\end{array} \right\}.
\]
%\[ 
%\mathcal{W}=\left\{ \beta(\uu_1,\xi_2\uu_2, \ldots, \xi_{k}\uu_k) \colon 
 %    \xi_h \in   \left(\sum\limits_{i=1}^{m-n_1}a'_{i} U_h\right)^{\perp'},  
 %     \mbox{ for } h \in \{2,\ldots,k\}
%\right\},\]
%where $U_{1}^{\perp'}=\langle a'_{1},\ldots,a'_{m-n_1}\rangle_{\F_q}$.
%In particular, the number of codewords $(\alpha_1 \uu_1,\ldots,\alpha_k\uu_k)$ of weight $n_1$, with $\alpha_1 \neq 0$ is \[(q^m-1)q^{j_2+\ldots+j_k},\] where $j_h= m-\dim_{\F_q}\left(\sum\limits_{i=1}^{m-n_1}a'_{i} U_h\right)$, for $h \in \{2,\ldots,k\}$.
\end{lemma}

\begin{proof}
Assume that $U_{t}^{\perp^*}=\langle a'_{1},\ldots,a'_{m-n_t}\rangle_{\F_q}$ and, by \eqref{eq:productsplitted}, we note that 
\[
 U_t^{\perp^*} U_h =\sum\limits_{i=1}^{m-n_t}a'_{i} U_h.
\]
Moreover, there exists a system associated with the code of the form $U=U_1 \times \ldots \times U_k$ and its dual is $U^{\perp'}=U_1^{\perp^*} \times \ldots \times U_k^{\perp^*}$, by \Cref{prop:dualcomplement}.
    First we prove that the codewords of $\C$ contained in $\mathcal{W}_t$ have weight equal to $n_t$. Let $\cc=\beta(\mathbf{0},\ldots,\mathbf{0},\uu_t,\xi_{t+1}\uu_{t+1}, \ldots, \xi_{k}\uu_k) \in \mathcal{W}_t$. Since $\xi_{h}  \in   \left(\sum\limits_{i=1}^{m-n_t}a'_{i} U_h\right)^{\perp^*}  = a_{1}'^{-1} U_h^{\perp^*}\cap \ldots \cap a_{m-n_t}'^{-1} U_h^{\perp^*}$, where the equality follows by \eqref{eq:productsplittedperp}, we get \begin{equation} \label{eq:writingxi}
\xi_h=w_{h,1}a_{1}'^{-1}=\ldots=w_{h,m-n_t}a_{m-n_t}'^{-1},
\end{equation}
for some $w_{h,1},\ldots,w_{h,m-n_1} \in U_h^{\perp^*}$, for any $h\in \{t+1,\ldots,k\}$. Let \[\mathbf{z}_j=(0,\ldots,0,a_j',w_{t+1,j},\ldots,w_{k,j} )\in U_1^{\perp^*} \times \ldots \times U_{k}^{\perp^*},\]
for $j \in \{1,\ldots,m-n_t\}$.
%This means that 
%\[
%\langle (a_1',w_{2,2},\ldots,w_{k,2} ) \rangle_{\F_{q^n}} = \ldots= \langle (a_{m-n_1}',w_{2,m-n_1},\ldots,w_{k,m-n_1}  ) \rangle_{\F_{q^n}}
%\]
%and since the vectors 
%\[
% (a_1',w_{2,2},\ldots,w_{k,2} ); \ldots= \langle (a_{m-n_1}',w_{2,m-n_1},\ldots,w_{k,m-n_1}  ) \rangle_{\F_{q^n}}
%\] 
Since $a_1',\ldots,a'_{m-n_t}$ are $\F_q$-linearly independent, we get that $\zz_1,\ldots,\zz_{m-n_t}$ are $\F_q$-linearly independent, as well, and by using \eqref{eq:writingxi}, we have $\langle \zz_1 \rangle_{\F_{q^m}}=\ldots=\langle \zz_{m-n_t} \rangle_{\F_{q^m}}$. It follows that 
\[\dim_{\F_q}((U_1^{\perp^*} \times \ldots \times U_{k}^{\perp^*} )\cap \langle \zz_1 \rangle_{\F_{q^m}}) \geq m-n_t.\]
Now, $\cc=\frac{\beta}{a_1'}\zz_1G$
and hence, $\w(\cc)=\w(\zz_1G) \leq n_t$, by \eqref{eq:relationweightdual}. Taking into account \eqref{eq:boundweightcompl}, we get that $\w(\cc)= n_t$.  
Conversely, let $\cc=(\alpha_1 \uu_1,\ldots,\alpha_k\uu_k)\in \Sigma(\C,t) \setminus \Sigma(\C,t+1)$ of weight $n_t$. Then $\alpha_i=0$, if $i<t$, $\alpha_t \neq 0$ and $\cc=\alpha_t(\mathbf{0},\ldots,\mathbf{0},\uu_t,\beta_{t+1}\uu_{t+1}, \ldots, \beta_k \uu_k)$, with $\beta_i=\alpha_i/\alpha_t$. So, $\cc=\alpha_t \zz G$, with $\zz=(0,\ldots,0,1,\beta_{t+1},\ldots,\beta_k)$. Since $n_t=\w(\cc)=\w(\zz G)$, then, by using \eqref{eq:relationweightdual}, we get $\dim_{\F_q}((U_t^{\perp^*} \times \ldots \times U_{k}^{\perp^*}) \cap \langle \zz \rangle_{\F_{q^m}})=m-n_t$. Therefore, there exist $m-n_t$ $\F_q$-linearly independent vectors such that
\[
\langle (0,\ldots,0,w_{t,1},\ldots,w_{k,1}),\ldots,(0,\ldots,0,w_{t,m-n_t},\ldots,w_{k,m-n_t}) \rangle_{\F_q} = (U_1^{\perp^*} \times \ldots \times U_{k}^{\perp^*}) \cap \langle \zz \rangle_{\F_{q^m}}.
\]
In particular, $w_{t,1},\ldots,w_{t,m-n_t} \neq 0$ and $\beta_i=w_{i,h}/ w_{t,h}$, for any $h \in \{t+1,\ldots,m-n_t \}$ and any $i \in \{t+1,\ldots,k\}$.
Note that, in this way, $w_{t,1},\ldots,w_{t,m-n_t}$ are $\F_q$-linearly independent.
Then $\langle w_{t,1},\ldots,w_{t,m-n_t} \rangle_{\fq}=U_t^{\perp'}$ and hence there exists a matrix $B=(b_{ij}) \in \mathrm{GL}(m-n_t,q)$ such that 
\[
B \left( \begin{matrix}
w_{t,1} \\
\vdots \\
w_{t,m-n_t}
\end{matrix}\right)=\left( \begin{matrix}
a_{1}' \\
\vdots \\
a_{m-n_t}'
\end{matrix}\right).
\]
For any $i \in \{t+1,\ldots,k\}$, define $\overline{w}_{i,1},\ldots,\overline{w}_{i,m-n_t} \in \F_{q^m}$ as 
\[\left( \begin{matrix}
\overline{w}_{i,1} \\
\vdots \\
\overline{w}_{i,m-n_t}
\end{matrix}\right)=B \left( \begin{matrix}
w_{i,1} \\
\vdots \\
w_{i,m-n_t}
\end{matrix}\right).
\]
Hence, $\overline{w}_{i,1},\ldots,\overline{w}_{i,m-n_t} \in U_i^{\perp'}$, for any $i \in \{t+1,\ldots,k\}$ and for $h\in \{1,\ldots,m-n_t\}$ \[\sum_{l=1}^{m-n_t} b_{hl} (0,\ldots,0,w_{t,l},\ldots,w_{k,l})=(0,\ldots,0,a_h',\overline{w}_{2,h},\ldots,\overline{w}_{k,h}) \in U_1^{\perp'} \times \ldots \times U_k^{\perp'}.\] This implies that \[(0,\ldots,0,a_h',\overline{w}_{2,h},\ldots,\overline{w}_{k,h}) \in  (U_1^{\perp^*} \times \ldots \times U_k^{\perp^*}) \cap \langle \zz \rangle_{\F_{q^m}},\] for any $h \in\{1,\ldots,m-n_t\}$. Let $\beta_h=\overline{w}_{h,i}/a_{i}'$. So $\beta_h \in a_{i}'^{-1}U_h^{\perp^*}$, for every $i\in\{t+1,\ldots,m-n_t\}$ providing that $\beta_{h}  \in  a_{1}'^{-1} U_h^{\perp^*}\cap \ldots \cap a_{m-n_t}'^{-1} U_h^{\perp^*}=\left(\sum\limits_{i=1}^{m-n_t}a'_{i} U_h\right)^{\perp^*}=\left( U_1^{\perp'} U_h \right)^{\perp^*}$, where the last two equalities follow by \eqref{eq:productsplittedperp} and \eqref{eq:productsplitted}, respectively. Therefore, $\cc \in \mathcal{W}_t$, that proves our assertion.
\end{proof}

Now, we are able to describe all the minimum weight codewords. For a subset $S \subseteq \F_{q^m}^n$ and a positive integer $i$, we denote by $A_i(S)$ the number of elements in $S$ having weight $i$.  

\begin{theorem} \label{th:numberpointsweightcomplementary}
Let $\C$ be a completely decomposable $\nk$ code of type $(n_1,\ldots,n_k)$ having as generator matrix $G=\uu_1\oplus \ldots \oplus \uu_k$ in a weight complementary form and let $U_i$ be the $\F_q$-span of the entries of $\uu_i$, for every $i$.
 Assume that 
$n_k=\ldots=n_{k-{\ell}} \neq n_{k-{\ell}-1}$, for some $\ell \in \{0,\ldots,k-1\}$, with $n_0:=0$. 
The number $A_{n_k}(\C)$ of codewords of weight $n_{k}$, is \[
A_{n_k}(\C)=\sum_{t=k-\ell}^k A_{n_k}\left(\Sigma(\C,t) \setminus \Sigma(\C,t+1)\right)=\begin{cases}
q^m-1 & \mbox{ if }\ell=0, \\
    (q^m-1)
    \left(\sum\limits_{i=k-\ell}^{k-1}\prod\limits_{h=i+1}^k q^{j_{i,h}}+1\right) & \mbox{ otherwise, }\end{cases}
    \] %where $j_{i,h}= m-\dim_{\F_q}\left(\sum\limits_{j=1}^{m-n_k}a'_{i,j} U_h\right)$, for every $i \in \{k-\ell,\ldots,k-1\}$ and $h \in \{i+1,\ldots,k\}$. 
    where $j_{i,h}= m-\dim_{\F_q}\left( U_i^{\perp^*} U_h\right)$, for every $i \in \{k-\ell,\ldots,k-1\}$ and $h \in \{i+1,\ldots,k\}$. 
\end{theorem}

\begin{proof}
   A codeword $\cc$ of $\C$ is of the form $\cc=(\alpha_1\uu_1,\ldots,\alpha \uu_k)$, for some $\alpha_1,\ldots,\alpha_k \in \F_{q^m}$. If $\cc$ has minimum weight $n_k$, by \eqref{eq:boundweightcompl}, we need to have $\cc \in \Sigma(\C,k-\ell)$. %and hence $\alpha_1=\ldots=\alpha_{k-\ell-1}=0$. 
   This implies that the codewords of $\C$ having weight $n_k$ are of the form \[(\mathbf{0},\ldots,\mathbf{0},\alpha_{k-\ell}\uu_{k-\ell},\ldots,\alpha_{k}\uu_k),\]
   for some $\alpha_{k-\ell},\ldots,\alpha_{k} \in \F_{q^m}$. %This means the number of codewords of $\C$ having weight $n_k$ is equal to the number $A'_{n_k}:=A_{n_k}(\C_{k-\ell}\oplus \ldots \oplus  \C_k)$ of codewords having weight $n_k$ of the code 
   %\[\C_{k-\ell}\oplus \ldots \oplus  \C_k=\{(\alpha_{k-\ell}\uu_{k-\ell},\ldots,\alpha_{k}\uu_k) \colon \alpha_i \in \F_{q^m}\}.\] So, in the next we determine this number $A_{n_k}'$. %If $\uu_i=(a_{i,1},\ldots,a_{i,n_k})$, then $U_i=\langle a_{i,1},\ldots,a_{i,n_k}\rangle_{\F_q}$ and let $U_i^{\perp'}=\langle a_{i,1}',\ldots,a_{i,m-n_k}'\rangle_{\F_q}$,
   %for every $i \in \{k-\ell,\ldots,k\}$. 
   Clearly, if $\ell=0$, then $A_{n_k}=q^m-1$. So, assume that $\ell\geq 1$. Denote by $\mathcal{W}_{t}$ the subset of $\Sigma(\C,t) \setminus \Sigma(\C,t+1)$ in which all the codewords having weight $n_k$. The sets $\mathcal{W}_{k-\ell},\ldots,\mathcal{W}_k$ determine a partition of the codewords of $\C$ having weight $n_k$, and since $ \lvert \mathcal{W}_t \rvert =A_{n_k}(\Sigma(\C,t) \setminus \Sigma(\C,t+1))$, we get 
   \[
 A_{n_k}(\C) =\sum_{t=k-\ell}^{k} \lvert \mathcal{W}_t \rvert =\sum_{t=k-\ell}^{k} A_{n_k}(\Sigma(\C,t) \setminus \Sigma(\C,t+1)).
   \]
   So in order to determine the quantity $A_{n_k}(\C)$, it is enough to compute $\lvert \mathcal{W}_t\rvert$, for any $t\in \{k-\ell,\ldots,k\}$. Clearly, $\lvert \mathcal{W}_k\rvert=q^m-1$. %Now, it is easy to see that the codewords of $\mathcal{W}_{t}$ correspond to the set of codewords $(\alpha_{i}\uu_{i},\ldots,\alpha_{k}\uu_k)$ of weight $n_k$, with $\alpha_{i} \neq 0$. 
   By Lemma \ref{lm:numberpointsweightcomplementary}, we get that
%\[ 
%\mathcal{W}_{k-\ell}=\left\{ \beta(\uu_{k-\ell},\xi_{k-\ell+1}\uu_{k-\ell+1}, \ldots, \xi_{k}\uu_k) \colon 
%     \xi_h \in   \left(\sum\limits_{i=1}^{m-n_k}a'_{i} U_h\right)^{\perp'},  
%      \mbox{ for } h \in \{k-\ell+1,\ldots,k\}
%\right\}.\]
%\[ 
%\mathcal{W}_{i}=\left\{ \beta(\mathbf{0},\ldots,\mathbf{0},\uu_{i},\xi_{i+1}\uu_{i+1}, \ldots, \xi_{k}\uu_k) \colon 
%     \xi_h \in   \left(\sum\limits_{j=1}^{m-n_k}a'_{i,j} U_h\right)^{\perp'},  
%      \mbox{ for } h \in \{i+1,\ldots,k\}
%\right\},\]
\[ 
\mathcal{W}_{t}=\left\{ \beta(\mathbf{0},\ldots,\mathbf{0},\uu_{t},\xi_{t+1}\uu_{t+1}, \ldots, \xi_{k}\uu_k) \colon 
     \xi_h \in   \left(U_t^{\perp^*} U_h\right)^{\perp^*},  
      \mbox{ for } h \in \{t+1,\ldots,k\}
\right\},\]
for every $t \in \{k-\ell,\ldots,k-1\}$.
As a consequence,
\[\lvert \mathcal{W}_{t} \rvert = (q^m-1)q^{j_{t,t+1}+\ldots+j_{t,k}},\] where $j_{t,h}= m-\dim_{\F_q}\left( U_t^{\perp^*} U_h\right)$, for every $t \in \{k-\ell,\ldots,k-1\}$ and $h \in \{t+1,\ldots,k\}$. %Now, we need to count the number of codewords of weight $n_k$ such that $\alpha_{k-\ell}=0$ and $\alpha_{k-\ell+1} \neq 0$. Arguing as before, we have that
%\[\lvert \mathcal{W}_{k-\ell+1} \rvert = (q^m-1)q^{j_{k-\ell+1,k-\ell+2}+\ldots+j_{k-\ell+1,k}},\] where $j_{k-\ell+1,h}= m-\dim_{\F_q}\left(\sum\limits_{i=1}^{m-n_k}a'_{k-\ell,i} U_h\right)$, for $h \in \{k-\ell+2,\ldots,k\}$. Using an inductive argument, we get
Finally, we have
\[
A_{n_k}(\C)=\sum_{t=k-\ell}^{k-1} \lvert \mathcal{W}_t \rvert +q^m-1= (q^m-1)\left(\sum_{t=k-\ell}^{k-1}q^{j_{t,t+1}+\ldots+j_{t,k}} +1\right),
\]
where $j_{t,h}= m-\dim_{\F_q}\left( U_t^{\perp^*} U_h\right)$, for $t \in \{k-\ell,\ldots,k-1\}$ and $h \in \{t+1,\ldots,k\}$. 
\end{proof}

\begin{comment}
\begin{remark} \Cref{lm:numberpointsweightcomplementary} and \Cref{th:numberpointsweightcomplementary} show that information on the weight distribution and codewords having certain weights of a completely decomposable code $\C=\bigoplus_{i=1}^k \C_i$ can be detected by the the $t$-shortened codes $\Sigma(\C,t)$.
\end{remark}
\end{comment}

As a consequence of \Cref{th:numberpointsweightcomplementary}, we derive upper and lower bound on the number of minimum weight codewords.

\begin{corollary} \label{cor:boundminimumnonprime}
Let $\C$ be a completely decomposable $\nk$ code of type $(n_1,\ldots,n_k)$. Assume that 
$n_k=\ldots=n_{k-{\ell}} \neq n_{k-{\ell}-1}$, for some $\ell \in \{0,\ldots,k-1\}$, with $n_0:=0$. 
Then
\[
(q^{m}-1)(\ell+1)  \leq A_{n_k}(\C) \leq  (q^m-1) \frac{q^{(\ell+1)(m-n_k)}-1}{q^{m-n_k}-1}.
\]
\end{corollary}

\begin{proof}
Up to equivalence, we can assume that $G=\uu_1\oplus \ldots \oplus \uu_k$ is a generator matrix in a complementary weight form for $\C$ and let $U_i$ be the $\F_q$-span of the entries of $\uu_i$, for every $i$. By using \Cref{th:numberpointsweightcomplementary}, we know that 
    \[
A_{n_k}=\begin{cases}
q^m-1 & \mbox{ if }\ell=0, \\
    (q^m-1)
    \left(\sum\limits_{i=k-\ell}^{k-1}\prod\limits_{h=i+1}^k q^{j_{i,h}}+1\right) & \mbox{ otherwise }\end{cases}
    \] %where $j_{i,h}= m-\dim_{\F_q}\left(\sum\limits_{j=1}^{m-n_k}a'_{i,j} U_h\right)$, for every $i \in \{k-\ell,\ldots,k-1\}$ and $h \in \{i+1,\ldots,k\}$. 
    where $j_{i,h}= m-\dim_{\F_q}\left( U_i^{\perp^*} U_h\right)$, for every $i \in \{k-\ell,\ldots,k-1\}$ and $h \in \{i+1,\ldots,k\}$. 
    Since $\dim_{\F_q}(U_i^{\perp^*}U_h) \geq \dim_{\F_q}(U_h) \geq n_k$, we have $0 \leq j_{i,h}\leq m-n_k$ and hence we get the upper bound on $A_{n_k}$ when $j_{i,h}=m-n_k$ and the lower bound when $j_{i,h}=0$, for any $i,h$.
\end{proof}

\subsection{Tightness of the bounds on minimum weight codewords}

In this subsection, we study the tightness of the bounds on the number of minimum weight codewords for a completely decomposable code provided in \Cref{cor:boundminimumnonprime}.

First, we note that the lower bound of Corollary \ref{cor:boundminimumnonprime} is sharp. Indeed, assume that $m=2e$ and let $\xi \in \F_{q^m} \setminus \F_{q^e}$. Consider, $\mu_1,\ldots,\mu_k \in \F_{q^e}$, where $k \leq q-1$, such that $\N_{q^e/q}(\mu_i)\ne\N_{q^e/q}(\mu_j)$ and $\N_{q^e/q}(\mu_i \mu_j \xi^{q^e+1})\ne 1$, for every $i\ne j$. Let $\lambda \in \F_{q^e}$ such that $\F_{q^e}=\F_q(\lambda)$.
The codes admitting a generator matrix of the form
\[ G=
\begin{pmatrix}
            \uu_1 & 0 & 0 & \cdots & 0 \\
            0 & \uu_2 & 0 & \cdots  & 0 \\
            0 & 0 & \ddots & \cdots & 0 \\
            \vdots & \vdots & \vdots & \ddots & \vdots \\
            0 & 0 & 0 & \cdots & \uu_k
        \end{pmatrix}
\in \F_{q^m}^{k\times ke}, \]
where $\uu_i=(1 +\xi \mu_i, \lambda + \xi \mu_i \lambda^q, \ldots, \lambda^{e-1} + \xi \mu_i(\lambda^{e-1})^q),$ for every $i \in \{1,\ldots,k\}$ are completely decomposable $[ke,k,e]_{q^m/q}$ codes of type $(e,\ldots,e)$ having $(q^m-1)k$ codewords of minimum weight $e$. See \cite[Section 6.2]{zullo2023multi} for more details on the weight distribution of such codes.
 %Denote by \[U_i=\{x+\xi \mu_i x^q \colon x \in \F_{q^e}\}^{\perp^*},\] for every $i \in \{1,\ldots,k\}$. Let $\uu_i \in \F_{q^m}^e$, such that the entries of $\uu_i$ forms a basis of $U_i$. Then the code
%\[
%\C=\bigoplus_{i=1}^k \langle \uu_i %\rangle_{\F_{q^m}},
%\]

The upper bound of \Cref{cor:boundminimumnonprime} is also sharp as proved in the following result.
\begin{proposition} \label{prop:constrlowernotprime}
    Let $m=re$, with $r>1$ and $e \geq 1$. Let $\xi \in \F_{q^m}$ such that $\F_{q^e}(\xi)=\F_{q^m}$. Define $\uu_i \in \F_{q^m}^{(r-1)e}$, such that the entries of $\uu_i$ forms an $\F_q$-basis of $\langle 1,\xi,\ldots,\xi^{r-2} \rangle_{\F_{q^e}}$. Let $\C_i=\langle \uu_i \rangle_{\F_{q^m}}$. Then the code
\[
\C=\bigoplus_{i=1}^k \C_i,
\]
is a completely decomposable $[k(r-1)e,k,(r-1)e]_{q^m/q}$ code of type $((r-1)e,\ldots,(r-1)e)$ having $(q^m-1)\frac{q^{ke}-1}{q^{e}-1}$ codewords of minimum weight $(r-1)e$ and all the other nonzero codewords have weight $m$.
\end{proposition}
\begin{proof}
Let $U_i$ be the $\F_q$-span of the entries of $\uu_i$ and note that $U_i = \langle 1,\xi,\ldots,\xi^{r-2} \rangle_{\F_{q^e}}$, for each $i$. By \Cref{th:numberpointsweightcomplementary}, we know that \begin{equation} \label{eq:minimumcodedubfield}
A_{e}(\C)=
    (q^m-1)
    \left(\sum\limits_{i=1}^{k-1}\prod\limits_{h=i+1}^k q^{j_{i,h}}+1\right)
    \end{equation}
    where $j_{i,h}= m-\dim_{\F_q}\left( U_i^{\perp^*} U_h\right)$, for every $i \in \{1,\ldots,k-1\}$ and $h \in \{i+1,\ldots,k\}$. So, in this case,
    \begin{equation} \label{eq:dualfqehyper}
U_i^{\perp^*}=(\langle 1,\xi,\ldots,\xi^{r-2} \rangle_{\F_{q^e}})^{\perp^*}=\zeta \F_{q^e},
    \end{equation}
  for some $\zeta \in \F_{q^m}^*$. Therefore,
    \[
    j_{i,h}= m-\dim_{\F_q}\left( U_i^{\perp^*} U_h\right) = m- (r-1)e=e.
    \]
    By substituting such values in \eqref{eq:minimumcodedubfield}, we get the number of minimum weight codewords. Finally, we show that all nonzero codewords have weight $m$ or $m-e=(r-1)e$, proving the assertion. Note that a system associated with $\C$ is $U=U_1 \times \cdots \times U_k$ and let $G$ be a generator matrix of $\C$ whose columns form and $\F_q$-basis for $U$. By using \eqref{eq:relationweightdual}, we get that a codeword $\xx G$ of $\C$ has weight $i$ if and only if 
$\dim_{\F_q}(U^{\perp'}\cap \langle \xx \rangle_{\F_{q^m}})=m-i$. Moreover, \Cref{prop:dualcomplement}, together with \eqref{eq:dualfqehyper}, imply that \[
U ^{\perp'} = U_1^{\perp^*} \times \ldots \times U_k^{\perp^*} =\zeta \F_{q^e} \times \ldots \times \zeta\F_{q^e}. \]
Therefore, $\dim_{\F_q}(U^{\perp'}\cap \langle \xx \rangle_{\F_{q^m}}) \in \{0,e\}$, implying $i\in \{m-e,m\}$.
\end{proof}

In the next, we characterize completely decomposable codes for which the number of minimum weight codewords attains the upper bound in \Cref{cor:boundminimumnonprime}.
To this aim, we need the following lemma, whose proof is straightforward.

\begin{lemma} \label{lem:powergeneralized}
Let $S$ be an $\fq$-subspace of $\F_{q^m}$ and let $a_2,\ldots,a_g \in \F_{q^m} ^*$. Assume that $\F_{q^e}=\fq(a_2,\ldots,a_g)$, where $\F_{q^e}$ is a subfield of $\F_{q^m}$. If $S=Sa_2=\ldots=Sa_g$, then $S$ is an $\F_{q^e}$-subspace of $\F_{q^m}$.
\end{lemma}

\begin{comment}
\begin{proof}
Up to a multiplication for a nonzero element of $S$, we can suppose that $1 \in S$. We have to show that all the fraction of polynomials with coefficients in $\F_q$ in the variables $a_2,\ldots,a_g$ belongs to S. Let $p(a_2,\ldots,a_g)=\frac{\sum_{i_2,\ldots,i_g}\alpha_{i_2,\ldots,i_g} a_2^{i_2} \cdots a_g^{i_g} }{\sum_{j_2,\ldots,j_g}\beta_{j_2,\ldots,j_g} a_2^{j_2} \cdots a_g^{j_g}}$. Since $S=Sa_i$ for each $i$, then $S=Sa_ia_i^{\ell-1}=Sa_i^{\ell}$ for any $\ell>0$. Hence $S a_2^{j_2} \cdots a_g^{j_g}=S=S a_2^{i_2} \cdots a_g^{i_g}$, for each $i_2,\ldots,i_g$ and $j_2,\ldots,j_g$. Moreover,  $S$ is also an $\F_q$-subspace and so we get
\[
S\left(\sum_{i_2,\ldots,i_g}\alpha_{i_2,\ldots,i_g} a_2^{i_2} \cdots a_g^{i_g}\right)=S=S\left( \sum_{j_2,\ldots,j_g}\beta_{j_2,\ldots,j_g} a_2^{j_2} \cdots a_g^{j_g}\right),
\]
from which the assertion follows.
\end{proof}
\end{comment}

The next result characterizes completely decomposable codes having the maximum number of minimum weight codewords.

\begin{theorem} \label{th:charminimumweightnotprime}
Let $\C$ be a completely decomposable $[n,k]_{q^m/q}$ code of type $(n_1,\ldots,n_k)$. Let
$n_k=\ldots=n_{k-{\ell}} \neq n_{k-{\ell}-1}$, for some $\ell \in \{1,\ldots,k-1\}$, with $n_0:=0$. If
\[A_{n_k}(\C) = (q^m-1) \frac{q^{(\ell+1)(m-n_k)}-1}{q^{m-n_k}-1} ,\]
then $m=re$, for some $r>1$ and $e \geq 1$, $n_k=(r-1)e$, and up to equivalence, 
  \[
    \Pi(\C,k-\ell)= \underbrace{\C' \oplus \ldots \oplus \C'}_{\ell+1},
    \]
where $\C'=\langle \uu \rangle_{\F_{q^m}}$, where the entries of $\uu$ forms an $\F_q$-basis of an $\F_{q^e}$-hyperplane of $\F_{q^m}$. In particular, if $\ell=k-1$, then up to equivalence, we have 
    \[
    \C=\bigoplus_{i=1}^{k} \C',
    \]
    with $\C'=\langle \uu \rangle_{\F_{q^m}}$, where the entries of $\uu$ form an $\F_q$-basis of an $\F_{q^e}$-hyperplane of $\F_{q^m}$.
\end{theorem}

\begin{proof}
Up to equivalence, there exists a generator matrix $G=\uu_1\oplus \ldots \oplus \uu_k$  of $\C$ in a weight complementary form. 
Let $U_i$ be the $\F_q$-span of the entries of $\uu_i$. The upper bound of \Cref{cor:boundminimumnonprime} is attained if and only if
\[
 m-n_k=j_{i,h}= m-\dim_{\F_q}\left( U_i^{\perp^*} U_h\right),
\]
or in other words
\begin{equation} \label{eq:equalityuppersubspace}
n_k=\dim_{\F_q}\left( U_i^{\perp^*} U_h\right),
\end{equation}
 for every $i \in \{k-\ell,\ldots,k-1\}$ and $h \in \{i+1,\ldots,k\}$. 
First, let consider $i=k-1$ and $h=k$ and the pair $(U_{k-1},U_k)$.
Let $a_1',\ldots,a_{m-n_k}'$ be an $\F_q$-basis of $U_{k-1}^{\perp^*}$. Since, $U_{k-1}^{\perp^*} U_k=a_1'U_k+\cdots+a_{m-n_k}'U_k$ and $\dim_{\F_q}(U_k)=n_k$, by \eqref{eq:equalityuppersubspace}, we get that
\[
a_1'U_k=\cdots=a_{m-n_k}'U_k.
\]
So, let $\F_{q^e}=\F_{q}(a_1',\ldots,a_{m-n_k}')$, where $\F_{q^e}$ is a subfield of $\F_{q^m}$ and so $m=er$, for some positive integer $r$. Clearly, $e \geq m-n_k$. By Lemma \ref{lem:powergeneralized}, we get that $U_k$ is an $\F_{q^e}$-subspace and as a consequence $e \mid n_k$ and $n_k=(r-1)e$. As a consequence, $U_{k-1}^{\perp^*}$ is an $\F_{q^e}$-subspace of dimension $1$. Hence, since $m=re$, $U_{k-1}$ is an $\F_{q^e}$-subspace of dimension $r-1$. So $U_{k-1}=c_{k-1}\langle 1,\ldots,\xi^{r-1} \rangle_{\F_{q^e}}$ and $U_{k}=c_{k}\langle 1,\ldots,\xi^{r-1} \rangle_{\F_{q^e}}$, for some $c_{k-1},c_k \in \F_{q^m}$ and $\xi \in \F_{q^m}$ such that $\F_{q^e}(\xi)=\F_{q^m}$, see e.g. \cite[Proposition 2.5]{napolitano2022clubs}. Arguing as before on the pairs $(U_i,U_h)$, for every $i \in \{k-\ell,\ldots,k-1\}$ and $h \in \{i+1,\ldots,k\}$, we get that
\[
U_i=c_i\langle 1,\ldots,\xi^{r-1} \rangle_{\F_{q^e}},
\]
for some $c_i \in \F_{q^m}$ and for every $i \in \{k-\ell,\ldots,k\}$.

So, a system associated with $\C$ is 
\[
U=U_1\times \ldots \times U_{k-\ell-1} \times c_{k-\ell}H \times \ldots c_{k-1}H \times c_kH, 
\]
where $H=\langle 1,\ldots,\xi^{r-1} \rangle_{\F_{q^e}}$ that is equivalent to 
\[
U'=U_1\times \ldots \times U_{k-\ell-1} \times H\times \ldots \times H.
\]
This means that, up to equivalence,  \[
    \C=\bigoplus_{i=1}^{k-\ell-1} \C_i \oplus \underbrace{\C' \oplus \ldots \oplus \C'}_{\ell+1},
    \]
    where $\C_i=\langle \uu_i \rangle_{\F_{q^m}}$, for some $\uu_i \in \F_{q^m}^{n_i}$, with $\w(\uu_i)=n_i$, for every $i \in \{1,\ldots,k-\ell-1\}$ and $\C'=\langle \uu \rangle_{\F_{q^m}}$, where the entries of $\uu$ forms an $\F_q$-basis of $\langle 1,\ldots,\xi^{r-1} \rangle_{\F_{q^e}}$. Therefore, $\Pi(\C,k-\ell)$ is as in the assertion.
\end{proof}

In the next result, we prove that a completely decomposable code with all blocks of the same length and the maximum number of minimum weight codewords is equivalent to the code constructed in \Cref{prop:constrlowernotprime}. In particular, its weight distribution is uniquely determined.

\begin{corollary} \label{cor:charmaximumminimnotprime}
    Let $\C$ be a completely decomposable $[n,k]_{q^m/q}$ code of type $(n_1,\ldots,n_k)$. Assume that
$n_1=\ldots=n_k$ and 
\[A_{n_k}(\C) = (q^m-1) \frac{q^{(\ell+1)(m-n_k)}-1}{q^{m-n_k}-1} .\]
Then $m=re$, for some $r>1$ and $e \geq 1$, $n_k=(r-1)e$, and up to equivalence, 
  we have 
    \[
    \C=\bigoplus_{i=1}^{k} \C',
    \]
    with $\C'=\langle \uu \rangle_{\F_{q^m}}$, where the entries of $\uu$ form an $\F_q$-basis of $\langle 1,\xi,\ldots,\xi^{r-2} \rangle_{\F_{q^e}}$, for some $\xi \in \F_{q^m}$ such that $\F_{q^e}(\xi)=\F_{q^m}$.  
\end{corollary}

\begin{proof}
By \Cref{th:charminimumweightnotprime}, we have that $m=re$, for some $r>1$ and $e \geq 1$, $n_k=(r-1)e$, and up to equivalence, 
    \[
    \C=\bigoplus_{i=1}^{k} \C',
    \]
    with $\C'=\langle \uu \rangle_{\F_{q^m}}$, where the entries of $\uu$ form an $\F_q$-basis of an $\F_{q^e}$-hyperplane of $\F_{q^m}$. Let $U'$ be the $\F_q$-span of the entries of $\uu$. Note that $U'$ is an $\F_{q^e}$-hyperplane of $\F_{q^m}$ and $U=U' \times \ldots \times U'$ is a system associated with $\C$. So, $U'=c\langle 1,\ldots,\xi^{r-1} \rangle_{\F_{q^e}}$, for some $c \in \F_{q^m}$ and $\xi \in \F_{q^m}$ such that $\F_{q^e}(\xi)=\F_{q^m}$, see e.g. \cite[Proposition 2.5]{napolitano2022clubs}. Then the system 
    \[
    \langle 1,\xi,\ldots,\xi^{r-2} \rangle_{\F_{q^e}} \times \ldots \times \langle 1,\xi,\ldots,\xi^{r-2} \rangle_{\F_{q^e}}
    \]
    is equivalent to $U$, implying the assertion.
\end{proof}

As proved in Theorem \ref{th:charminimumweightnotprime}, it is possible to construct codes whose number of minimum weight codewords attains the upper bound in \Cref{cor:boundminimumnonprime}, only when $m=re$ and $n_k=(r-1)e$ and so they have a common divisor $e$. Such a bound is not tight when the degree extension $m$ is a prime. %Indeed, for prime degree extension of $\F_q$, the following upper bound on $A_{n_k}(\C)$ holds.

\begin{corollary} \label{cor:boundminimumprime}
Assume that $m$ is a prime. Let $\C$ be a completely decomposable $\nk$ code of type $(n_1,\ldots,n_k)$. Assume that 
$n_k=\ldots=n_{k-{\ell}} \neq n_{k-{\ell}-1}$, for some $\ell \in \{0,\ldots,k-1\}$, with $n_0:=0$. Then 
\begin{equation} \label{eq:upperboundminimumprime}
 A_{n_k}(\C) \leq (q^m-1) \frac{q^{\ell+1}-1}{q-1} .
\end{equation}
\end{corollary}

\begin{proof}
Up to equivalence, assume that $G=\uu_1\oplus \ldots \oplus \uu_k$ is a generator matrix in a weight complementary form for $\C$ and let $U_i$ be the $\F_q$-span of the entries of $\uu_i$, for every $i$. By using \Cref{th:numberpointsweightcomplementary}, we know that 
    \[
A_{n_k}(\C)=\begin{cases}
q^m-1 & \mbox{ if }\ell=0, \\
    (q^m-1)
    \left(\sum\limits_{i=k-\ell}^{k-1}\prod\limits_{h=i+1}^k q^{j_{i,h}} +1\right) & \mbox{ otherwise }\end{cases}
    \] %where $j_{i,h}= m-\dim_{\F_q}\left(\sum\limits_{j=1}^{m-n_k}a'_{i,j} U_h\right)$, for every $i \in \{k-\ell,\ldots,k-1\}$ and $h \in \{i+1,\ldots,k\}$. 
    where $j_{i,h}= m-\dim_{\F_q}\left( U_i^{\perp^*} U_h\right)$, for every $i \in \{k-\ell,\ldots,k-1\}$ and $h \in \{i+1,\ldots,k\}$. By hypotheses, we have that $\dim_{\F_q}(U_i^{\perp^*})=m-n_k$, for any $i \in \{k-\ell,\ldots,k-1\}$, implying that, by Theorem  \ref{teo:bachocserrazemor},
    \[
    \dim_{\F_q}\left( U_i^{\perp^*} U_h\right) \geq m-n_k+n_k-1=m-1,
    \]
    and so $j_{i,h} \leq 1$, for every $i$ and $h$. Then, the assertion follows.
\end{proof}

By the above result, when $m$ is prime, the quantity $(q^m-1) \frac{q^{\ell+1}-1}{q-1}$ is roughly $q^{m+\ell}$ in \eqref{eq:upperboundminimumprime} provides an upper bound on $A_e(\C)$, significantly reducing the previous upper bound in \Cref{cor:boundminimumnonprime} defined by the quantity $(q^m-1) \frac{q^{(\ell+1)(m-n_k)}-1}{q^{m-n_k}-1}$ that is roughly $q^{m+\ell(m-n_k)}$.

\begin{remark} \label{rk:equalityweightminimum}
    Note that by the above theorem, we also deduce that the equality holds in \eqref{eq:upperboundminimumprime}, i.e.
    \[
    (q^m-1) \frac{q^{\ell+1}-1}{q-1} = A_{n_k}(\C)
    \]
    if and only if $j_{i,h}=1$, or equivalently $\dim_{\F_q}(U_i^{\perp^*}U_h)=m-1$, for every $i \in \{k-\ell,\ldots,k-1\}$ and $h \in \{i+1,\ldots,k\}$.
\end{remark}

For a prime degree extension, we characterize completely decomposable code $\C$ such that $A_{n_k}(\C)$ reaches the maximum in \eqref{eq:upperboundminimumprime}. % then some of its punctured codes are determined.

%We are the prove that completely decomposable codes $\C$ with the same number of codewords having minimum weight, need to have a particular structure.

\begin{theorem} \label{th:charshortnedpunct}
Assume that $m$ is prime. Let $\C$ be a completely decomposable $[n,k]_{q^m/q}$ code of type $\mathbf{n}=(n_1,\ldots,n_k)$. Let
$n_k=\ldots=n_{k-{\ell}} \neq n_{k-{\ell}-1}$, for some $\ell \in \{1,\ldots,k-1\}$, with $n_0:=0$. Then
\[A_{n_k}(\C) = (q^m-1) \frac{q^{\ell+1}-1}{q-1} ,\] if and only if, up to equivalence, we have 
    \[
    \Pi(\C,k-\ell)= \underbrace{\C' \oplus \ldots \oplus \C'}_{\ell+1},
    \]
with $\C'=\langle \uu \rangle_{\F_{q^m}}$, for some $\uu=(u_1,\ldots,u_{n_k}) \in \F_{q^m}^{n_k}$, where $\w(\uu)=n_k$ and $\dim_{\F_q}(U^{\perp^*}U)=m-1$, with $U=\langle u_1,\ldots,u_{n_k}\rangle_{\F_q}$.
%In particular, if $\uu_k=(1,\lambda,\ldots,\lambda^{n_k-1})$, for some $\lambda \in \F_{q^m} \setminus \F_q$, up to equivalence, we have 
 %   \[
 %   \C=\bigoplus_{i=1}^{k-\ell-1} \C_i \oplus \underbrace{\C_{\lambda,n_{k}} \oplus \ldots \oplus \C_{\lambda,n_{k}}}_{\ell+1},
%    \]
%    for some where $\C_i=\langle \uu_i \rangle_{\F_{q^m}}$, for $i \in \{1,\ldots,k-\ell-1\}$.
\end{theorem}

\begin{proof}
Up to equivalence, we can assume that $\C$ has a generator matrix $G=\uu_1\oplus \ldots \oplus \uu_k$ in a weight complementary form. Let $U_i$ be the $\F_q$-span of the entries of $\uu_i$, for every $i$. By Remark \ref{rk:equalityweightminimum}, we have that $A_{n_k}(\C) = (q^m-1) \frac{q^{\ell+1}-1}{q-1}$ if and only if
\[
\dim_{\F_q}( U_{i}^{\perp^*}U_{k} )=m-1,
\]
for every $i\in \{k-\ell,\ldots,k-1\}$. Moreover, by hypotheses, $\dim_{\F_q}(U_i^{\perp^*})=m-\dim_{\F_q}(U_k)$ holds, Therefore, by using Lemma \ref{lem:gencriticalm-1}, we have that $U_i=d_iU_k$, for certain $d_i \in \F_{q^m}^*$. So, a system associated with $\C$ is 
\[
U=U_1\times \ldots \times U_{k-\ell-1} \times d_{k-\ell}U_{k} \times \ldots d_{k-1}U_{k} \times U_{k}, 
\]
that is equivalent to 
\[
U'=U_1\times \ldots \times U_{k-\ell-1} \times U_{k} \times \ldots \times U_{k}.
\]
This implies that, up to equivalence, 
\[
    \C=\bigoplus_{i=1}^{k-\ell-1} \C_i \oplus \underbrace{\C' \oplus \ldots \oplus \C'}_{\ell+1},
    \]
    where $\C_i=\langle \uu_i \rangle_{\F_{q^m}}$, for $i \in \{1,\ldots,k-\ell-1\}$ and $\C'=\langle \uu_k \rangle_{\F_{q^m}}$, such that $\dim_{\F_q}(U_k^{\perp^*}U_k)=m-1$, from which we get the assertion.% Finally, if $\uu_k=\uu_{\lambda,n_k}$, we have that $U_k=U_{\lambda,n_k}$ and \[
%U'=U_1\times U_{k-\ell-1} \times U_{\lambda,n_{k-\ell}} \times U_{\lambda,n_{k-1}} \times U_{\lambda,n_k},
%\]
%is a system associated to a code equivalent to $\C$, from which we get the last part of the assertion. 

%By hypotheses, $U_k=U_{\lambda,n_k}$ and since $\dim_{\F_q}(U_i)=m-n_k$, for every $i\in \{k-\ell,\ldots,k-1\}$, by 2. of Theorem \ref{teo:bachocserrazemor}, we get that $U_i^{\perp'}=c_iU_{\lambda,m-n_i}$, for some $c_i \in \F_{q^m}$. Hence, \Cref{prop:dualwithdual} implies that $U_i=d_iU_{\lambda,n_{i}}$, for some $d_i \in \F_{q^m}$. So, 
%\[
%U=U_1\times U_{k-\ell-1} \times d_{k-\ell}U_{\lambda,n_{k-\ell}} \times \ldots d_{k-1}U_{\lambda,n_{k-1}} \times U_{\lambda,n_k}, 
%\]
%that is equivalent to 
%\[
%U'=U_1\times U_{k-\ell-1} \times U_{\lambda,n_{k-\ell}} \times U_{\lambda,n_{k-1}} \times U_{\lambda,n_k},
%\]
%from which we get the assertion. 
\end{proof}

In particular, if the $n_i$'s are equal, every completely decomposable code, having the maximum possible number of minimum weight codewords, is uniquely determined by one element $\uu \in \F_{q^m}^{n_i}$.

\begin{corollary} \label{cor:maximumnumeberprime}
    Assume that $m$ is prime. Let $\C$ be a completely decomposable $[n,k]_{q^m/q}$ code of type $(n_1,\ldots,n_k)$ with $n_1=\ldots=n_k$. Then
\[A_{n_k}(\C) = (q^m-1) \frac{q^{k}-1}{q-1} ,\] if and only if, up to equivalence, we have 
    \[
    \C= \bigoplus_{i_1}^k \C',
    \]
where $\C'=\langle \uu \rangle_{\F_{q^m}}$, for some $\uu=(u_1,\ldots,u_{n_k}) \in \F_{q^m}^{n_k}$, with $\w(\uu)=n_k$ and $\dim_{\F_q}(U^{\perp^*}U)=m-1$, with $U=\langle u_1,\ldots,u_{n_k}\rangle_{\F_q}$.
\end{corollary}

\begin{proof}
    It immediately follows by \Cref{th:charshortnedpunct}, for the case $\ell=k-1$. 
\end{proof}

We introduce a family of completely decomposable codes having as number of minimum weight codewords as in \Cref{th:charshortnedpunct}. 
For an element $\lambda \in \F_{q^m}$ having degree $e$ over $\F_q$, and a positive integer $t\leq e$, define $\uu_{\lambda,e,t}=(1,\lambda,\ldots,\lambda^{t-1})$ and $\C_{\lambda,e,t}$ the nondegenerate $[t,1]_{q^m/q}$ code generated by $\uu_{\lambda,e,t}$. Moreover, we denote the $\F_q$-subspace generated by the entries of $\uu_{\lambda,e,t}$ by $U_{\lambda,e,t}$, i.e.  \[U_{\lambda,e,t}=\langle 1,\lambda,\ldots,\lambda^{t-1}\rangle_{\F_q}.\] When the degree of $\lambda$ equals to $m$, or in other words, $\lambda$ is a generator of $\F_{q^m}$ of $\F_q$, we simply write $\uu_{\lambda,t}, \C_{\lambda,t}$ and $ U_{\lambda,t},$  instead of $\uu_{\lambda,m,t},\C_{\lambda,m,t}$ and $U_{\lambda,m,t}$, respectively.

\begin{proposition} \label{prop:numberminimumweight} Let
$n_k=\ldots=n_{k-{\ell}} \neq n_{k-{\ell}-1}$, for some $\ell \in \{1,\ldots,k-1\}$, with $n_0:=0$. For $e \geq n_{k-\ell}$ and $e \mid m$, the code
\[
\C=\bigoplus_{i=1}^{k-\ell-1}\C_i \oplus \C_{\lambda,e,n_{k-\ell}} \oplus \cdots \oplus \C_{\lambda,e,n_{k}},
\]
where $\C_i$ is an MRD $[n_i,1]_{q^m/q}$ code, for every $i\in \{1,\ldots,k-\ell-1\}$, is a completely decomposable code of type $(n_1,\ldots,n_k)$ with \[A_{n_k}(\C)=(q^m-1)\frac{q^{\ell+1}-1}{q-1}.\] In particular, the value $A_{n_k}(\C)$ does not depend on the choice of $\lambda$.
\end{proposition}

\begin{proof}
    %The code $\C_{\lambda,e,\mathbf{n}}$ has a generator matrix $G=\uu_1\oplus \ldots \oplus \uu_{k-\ell-1} \oplus \uu_{\lambda,e,n_{k-\ell}} \oplus \ldots \oplus \uu_{\lambda,e,{n_k}}$, for some $\uu_i \in \F_{q^m}^{n_i}$, with $i\in \{1,\ldots,k-\ell-1\}$.  
    If $\ell=0$, the assertion is trivial. So assume that $\ell>0$. Let $t=n_{k-\ell}=\cdots=n_k$. Note that $\C_{\lambda,e,t}=\langle \uu_{\lambda,e,t}\rangle_{\F_{q^m}}$ and the $\F_q$-span of the entries of $\uu_{\lambda,e,t}$ is $U_{\lambda,e,t}$. Therefore, by using \Cref{th:charshortnedpunct}, we get the assertion if we prove that 
    \[
    \dim_{\F_q}(U_{\lambda,e,t}^{\perp ^*}U_{\lambda,e,t})=m-1.    \]
Assume first that $e=m$. By using Proposition \ref{prop:dualwithdual}, for every $i\in \{k-\ell,\ldots,k\}$, we get that \[U_{\lambda,t}^{\perp^*}=c\langle 1,\lambda,\ldots,\lambda^{m-t-1} \rangle_{\F_q},\] for some nonzero $c \in \F_{q^m}^*$. Therefore,
    \[
    U_{\lambda,t}^{\perp^*} U_{\lambda,t}=c\langle 1,\lambda,\ldots,\lambda^{m-2}\rangle_{\F_q}
    \]
    and $\dim_{\F_q}(U_{\lambda,t}^{\perp^*} U_{\lambda,t})=m-1$.\\
    Assume now that $m=es$, for some $e,s>1$. By using Proposition \ref{prop:duallambdasubfield}, we get that \[U_{\lambda,e,t}^{\perp^*}=\mathrm{Ker}(\mathrm{Tr}_{q^m/q^e})\oplus c U_{\lambda,e,e-t},
    \]
    for some $c \in \F_{q^m}^*$. Therefore, since $\mathrm{Ker}(\mathrm{Tr}_{q^m/q^e})$ is an $\F_{q^e}$-subspace of $\F_{q^m}$, we get
    \[
    U_{\lambda,e,t}^{\perp^*} U_{\lambda,e,t}=\mathrm{Ker}(\mathrm{Tr}_{q^m/q^e}) \oplus cU_{\lambda,e,e-1},
    \]
implying that $\dim_{\F_q}(U_{\lambda,e,t}^{\perp^*} U_{\lambda,e,t})=e(s-1)+e-1=es-1=m-1$, form which we get the assertion. 
\end{proof}

Note that, when $m$ is prime, the quantity $A_{n_k}(\C)$ reaches the maximum in \eqref{eq:upperboundminimumprime}.

\begin{remark}
    Under certain assumptions on the $n_i$'s, the codes 
    \[\C=\bigoplus_{i=1}^k \C_{\lambda,n_i}\]
    reach the equality on the upper bound for maximum weight codewords of \cite[Theorem 5.8]{polverino2023maximum}.
\end{remark}

By \Cref{prop:numberminimumweight}, we also get that the number of codewords having minimum weight $n_k$ does not depend on the choice of $\lambda$. However, different choices of $\lambda$ can produce codes with different weight distribution. 

\begin{example}
     Let $q=2$, $m=6$ and $e \mid m$. Define $\mathbf{n}=(2,2,2)$. As proved in \Cref{prop:numberminimumweight}, the number of codewords of the code
     $\C_{\lambda}:=\bigoplus_{i=1}^3 \C_{\lambda,e,2}$,
     having minimum weight 2 does not depend on $\lambda$ and
     \[
     A_2(\C_{\lambda})=(2^6-1) \frac{2^3-1}{2-1}=441,
     \]
     for any $\lambda \in \F_{2^6}$, where $e$ is such that $\F_{2^e}=\F_2(\lambda)$. However,
     by using MAGMA, it is possible to show that there exist $\lambda_1 \in \F_{2^6}$ such that $\F_2(\lambda_1)=\F_{2^6}$ and $\lambda_2 \in \F_{2^3}$ such that $\F_2(\lambda_2)=\F_{2^3}$ for which the weight distribution of $\C_{\lambda_1}$ is \[
(A_0,A_1,A_2,A_3,A_4,A_5,A_6 )=(1,0,441,2646,35280,127008,96768).
     \]
     and the weight distribution of $\C_{\lambda_2}$ is 
     \[
(A'_0,A'_1,A'_2,A'_3,A'_4,A'_5,A'_6 )=(1,0,441,4158,24696,148176,84672).
     \]
     Therefore, the completely decomposable $[6,3,2]_{2^6/2}$ codes $\C_{\lambda_1}$ and $\C_{\lambda_2}$ of type $(2,2,2)$ have a different weight distribution.
\end{example}

Finally, in Corollary \ref{cor:charmaximumminimnotprime}, we have showed that a completely decomposable code with all blocks of the same length and the maximum number of minimum weight codewords as in \Cref{cor:boundminimumnonprime} is equivalent to the code constructed in \Cref{prop:constrlowernotprime} and hence its weight distribution is uniquely determined. In the next example, we prove that for the case in which $m$ is prime, there exist completely decomposable codes having the maximum number of minimum weight codewords as in \Cref{cor:boundminimumprime}, but having different weight distributions.

\begin{example}
     Let $q=2$ and $m=7$. Define $(n_1,n_2,n_3)=(3,3,3)$ and consider the $[9,3,3]_{2^7/2}$ code $\C_{\lambda}:=\bigoplus_{i=1}^3 \C_{\lambda,n_i}$, for some $\lambda \in \F_{2^7} \setminus \F_2$. By \Cref{prop:numberminimumweight}, we know that the number of codewords of $\C_{\lambda}$ having minimum weight 3 does not depend on $\lambda$ and
     \[
     A_3(\C_{\lambda})=(2^7-1) \frac{2^3-1}{2-1}=127 \cdot 7=889.
     \]
     However, by using MAGMA, we get that there exists $\lambda \in \F_{2^7}$ such that the weight distribution of $\C_{\lambda}$ is \[
(A_0,A_1,A_2,A_3,A_4,A_5,A_6,A_7,A_8,A_9)=(1,0,0,127\cdot 7, 127\cdot 42,127 \cdot 336, 127 \cdot 2688, 127 \cdot 13440, 0, 0).
     \]
     and the weight distribution of the code $\C'=\bigoplus_{i=1}^3 \langle \uu \rangle_{\F_{2^7}}$, where $\uu=(1,\lambda,\lambda^3)$ is 
     \[
(A'_0,A'_1,A'_2,A'_3,A'_4,A'_5,A'_6,A'_7,A'_8,A'_9)=(1,0,0,127\cdot 7, 0,127 \cdot 294, 127 \cdot 3108, 127 \cdot 13104, 0, 0).
     \]
    Therefore, the completely decomposable $[9,3,3]_{2^7/2}$ codes $\C_{\lambda}$ and $\C'$ of type $(3,3,3)$ have different weight distributions, although $m$ is prime and they have the same number of codewords with the minimum weight 3.
\end{example}

\section*{\hskip 5.5 cm Acknowledgements}

The author is very grateful to Gianira N. Alfarano for fruitful discussions on minimal codewords of a completely decomposable code.
The research was supported by the project ``COMBINE'' of the University of Campania ``Luigi Vanvitelli'' and was partially supported by the Italian National Group for Algebraic and Geometric Structures and their Applications (GNSAGA - INdAM) and by the INdAM - GNSAGA Project \emph{Tensors over finite fields and their applications}, number E53C23001670001.

\bibliographystyle{abbrv}
\bibliography{biblio}

Paolo Santonastaso\\
Dipartimento di Matematica e Fisica,\\
Universit\`a degli Studi della Campania ``Luigi Vanvitelli'',\\
I--\,81100 Caserta, Italy\\
{{\em paolo.santonastaso@unicampania.it}}

\end{document}